\batchmode
\documentclass[twoside,12pt]{article}
\usepackage{color}
\usepackage{url}
\usepackage{enumitem}
\usepackage{amsmath}
\usepackage{amsthm}
\usepackage{amssymb}
\usepackage[all]{xy}
\PassOptionsToPackage{normalem}{ulem}
\usepackage{ulem}
\usepackage[unicode=true,
 bookmarks=true,bookmarksnumbered=false,bookmarksopen=false,
 breaklinks=false,pdfborder={0 0 0},pdfborderstyle={},backref=page,colorlinks=true]
 {hyperref}
\hypersetup{pdftitle={Analyzing Games with Ambiguous Player Types Using the MINthenMAX Decision Model},
 pdfauthor={Ilan Nehama}}

\makeatletter

\usepackage[all]{xy}
\usepackage{fullpage}

\usepackage{csquotes}
\usepackage{xspace}

\setlist[itemize,enumerate]{nosep}

\global\long\def\MIN{\text{{\rm MIN}}}
\global\long\def\MAX{\text{{\rm MAX}}}
\global\long\def\LEX{\text{{\rm \ensuremath{\MIN}then\ensuremath{\MAX}}}}

\global\long\def\MINBR{\text{{\rm \ensuremath{\MIN}-BR}}}
\global\long\def\LEXBR{\text{{\rm \ensuremath{\LEX}-BR}}}
\global\long\def\BR{\text{{\rm BR}}}

\global\long\def\MINNE{\text{{\rm \ensuremath{\MIN}-NE}}}
\global\long\def\LEXNE{\text{{\rm \ensuremath{\LEX}-NE}}}
\global\long\def\NE{\text{{\rm NE}}}

\global\long\def\ActionSet{\mathcal{A}}

\global\long\def\TypeSet{\mathcal{T}}

\renewcommand\MINBR{\texify{\MIN-BR}\xspace}
\renewcommand\MINNE{\texify{\MIN-NE}\xspace}
\renewcommand\LEXBR{\texify{\LEX-BR}\xspace}
\renewcommand\LEXNE{\texify{\LEX-NE}\xspace}

\renewcommand\LEX{\texify{\texttt{MIN\-then\-MAX}}\xspace}
\renewcommand\LEX{\texify{\scshape{MIN\-then\-MAX}}\xspace}
\renewcommand\MIN{\texify{\scshape{MIN}}\xspace}

\ifx\NEWCOMMANDS\undefined
\newcommand{\NEWCOMMANDS}{}

\newcommand{\mathify}[1]{\ifmmode{#1}\else\mbox{\ensuremath{#1}}\fi} 
\newcommand{\texify}[1] {\ifmmode{\text{#1}}\else{#1}\fi}

\global\long\def\st{\text{ s.t. }}

\global\long\def\Exp{\mathbb{E}}
\global\long\def\Exp{\mathop{\mathbb{E}}}

\global\long\def\half{\frac{1}{2}}

\global\long\def\argmax{\mathop{\mathrm{argmax}}}

\global\long\def\argmin{\mathop{\mathrm{argmin}}}

\makeatletter

\global\long\def\Rmnum#1{\uppercase\expandafter{\romannumeral#1\relax}}

\makeatother

\global\long\def\SetSt{\;\middle\vert\;}

\global\long\def\sizeof#1{\left|#1\right|}

\newcommand{\XXth}[2]{%
		\ifmmode {#1^{\underline{#2}}}%
		\else {#1\textsuperscript{\underline{#2}}\xspace}\fi%
}
\newcommand{\Xth}[1]{\XXth{#1}{th}\xspace}

\newcommand{\Nst}[1]{\mathify{\XXth{#1}{st}}\xspace}
\newcommand{\Nnd}[1]{\mathify{\XXth{#1}{nd}}\xspace}
\newcommand{\Nrd}[1]{\mathify{\XXth{#1}{rd}}\xspace}
\newcommand{\Nth}[1]{%
\ifcase#1\relax%
		0
    \or\Nst{1}%
    \or\Nnd{2}%
    \or\Nrd{3}%
    \else\mathify{\Xth{#1}}%
\fi\xspace}

\global\long\def\Players{\ensuremath{\mathcal{N}}}

\global\long\def\prefOver{\succcurlyeq}

\global\long\def\strongPrefOver{\succ}
\global\long\def\SprefOver{\strongPrefOver}

\global\long\def\prefBy{\preccurlyeq}

\global\long\def\strongPrefBy{\prec}
\global\long\def\SprefBy{\strongPrefBy}

\newcommand{\ProblemClass}[1]{\texify{\textit{#1}}}
\newcommand{\ProblemClassComplete}[1]{\ProblemClass{#1-complete}\xspace} 
\newcommand{\ProblemClassHard}[1]{\ProblemClass{#1-hard}\xspace} 

\fi

\usepackage{xfrac} 
\usepackage{slashbox} 
\usepackage{xfrac} 
\usepackage{slashbox} 
\providecommand{\tabularnewline}{\\}

\theoremstyle{plain}
\newtheorem{thm}{\protect\theoremname}
  \theoremstyle{definition}
  \newtheorem{defn}[thm]{\protect\definitionname}
  \theoremstyle{plain}
  \newtheorem{lem}[thm]{\protect\lemmaname}
\newenvironment{proof-sketch}[0]
{\begin{proof}[\textbf{Proof sketch}]}
{\end{proof}}
  \theoremstyle{plain}
  \newtheorem{cor}[thm]{\protect\corollaryname}
  \theoremstyle{plain}
  \newtheorem{ax}[thm]{\protect\axiomname}
\newenvironment{proof-of-something}[2]
{\begin{proof}[\textbf{Proof of #2~#1}]}
{\end{proof}}
  \theoremstyle{plain}
  \newtheorem{prop}[thm]{\protect\propositionname}
  \theoremstyle{definition}
  \newtheorem{example}[thm]{\protect\examplename}
  \theoremstyle{plain}
  \newtheorem*{lem*}{\protect\lemmaname}
  \theoremstyle{plain}
  \newtheorem*{thm*}{\protect\theoremname}
  \theoremstyle{plain}
  \newtheorem*{ax*}{\protect\axiomname}
  \theoremstyle{plain}
  \newtheorem*{prop*}{\protect\propositionname}
\newcommand{\MyLyxThmNewline}[1]{~\par\nopagebreak\ignorespaces}
\newcommand{\VspaceBeforeMath}[2]{\vspace{#1}}
\newcommand{\MyLyxThmFN}[1]{\hspace{-6pt}\footnote{#1}}

\makeatother

  \providecommand{\axiomname}{Axiom}
  \providecommand{\corollaryname}{Corollary}
  \providecommand{\definitionname}{Definition}
  \providecommand{\examplename}{Example}
  \providecommand{\lemmaname}{Lemma}
  \providecommand{\propositionname}{Proposition}
  \providecommand{\theoremname}{Theorem}
\providecommand{\theoremname}{Theorem}

\begin{document}

\title{Analyzing Games with Ambiguous Player Types Using the $\LEX$ Decision
Model\\
{\large{}(work in progress)}}

\author{Ilan Nehama\\
Tel-Aviv University\\
{\normalsize{}\url{ilan.nehama@mail.huji.ac.il}}}

\date{November 2016}

\maketitle
\thispagestyle{empty}
\begin{abstract}

\global\long\def\MIN{\text{{\rm MIN}}}
\global\long\def\MAX{\text{{\rm MAX}}}
\global\long\def\LEX{\text{{\rm \ensuremath{\MIN}then\ensuremath{\MAX}}}}

\global\long\def\MINBR{\text{{\rm \ensuremath{\MIN}-BR}}}
\global\long\def\LEXBR{\text{{\rm \ensuremath{\LEX}-BR}}}
\global\long\def\BR{\text{{\rm BR}}}

\global\long\def\MINNE{\text{{\rm \ensuremath{\MIN}-NE}}}
\global\long\def\LEXNE{\text{{\rm \ensuremath{\LEX}-NE}}}
\global\long\def\NE{\text{{\rm NE}}}

\global\long\def\ActionSet{\mathcal{A}}

\global\long\def\TypeSet{\mathcal{T}}

\renewcommand\MINBR{\texify{\MIN-BR}\xspace}
\renewcommand\MINNE{\texify{\MIN-NE}\xspace}
\renewcommand\LEXBR{\texify{\LEX-BR}\xspace}
\renewcommand\LEXNE{\texify{\LEX-NE}\xspace}

\renewcommand\LEX{\texify{\texttt{MIN\-then\-MAX}}\xspace}
\renewcommand\LEX{\texify{\scshape{MIN\-then\-MAX}}\xspace}
\renewcommand\MIN{\texify{\scshape{MIN}}\xspace}



In many common interactive scenarios, participants lack information
about other participants, and specifically about the preferences of
other participants. In this work, we model an extreme case of  incomplete
information, which we term \emph{games with type ambiguity}, where
a participant lacks even information enabling him to form a belief
on the preferences of others. Under type ambiguity, one cannot analyze
the scenario using the commonly used Bayesian framework,  and therefore
one needs to model the participants using a different decision model.

To this end, we present the $\LEX$ decision model under ambiguity.
This model is a refinement of Wald's MiniMax principle, which we show
to be too coarse for games with type ambiguity. We characterize $\LEX$
as the finest refinement of the MiniMax principle that satisfies three
properties we claim are necessary for games with type ambiguity. This
prior-less approach we present here also follows the common practice
in computer science of worst-case analysis. 

Finally, we define and analyze the corresponding equilibrium concept,
when all players follow $\LEX$. We demonstrate this equilibrium by
applying it to two common economic scenarios: coordination games and
bilateral trade. We show that in both scenarios, an equilibrium in
pure strategies always exists, and we analyze the equilibria.

\end{abstract}
\textbf{\uline{Keywords}}: Decision making under ambiguity, Games
with ambiguity, Wald's MiniMax principle, $\LEX$ decision model,
$\LEX$ equilibrium 

{\renewcommand{\thefootnote}{}\footnote{This work is supported by ISF grant 1435/14 administered by the Israeli
Academy of Sciences and Israel-USA Binational Science Foundation (BSF)
grant 2014389.}\footnote{This work was carried out primarily when the author was a Ph.D. student
at the Hebrew University of Jerusalem, Israel.}\footnote{
The author would like to thank Tomer Siedner, Yannai Gonczarowski,
Gali Noti, Shmuel Zamir, Bezalel Peleg, Noam Nisan,  Mike Borns,
as well as participants  of Social Choice and Welfare 2016, Games
2016, and the Computation and Economics seminars at the Hebrew University
and Tel-Aviv University for numerous discussions and comments that
helped to improve this paper.

}\setcounter{footnote}{0}}

\newpage{}

\tableofcontents{}\setcounter{page}{0}
\thispagestyle{empty}

\newpage{}

\section{Introduction}

In many common interactive scenarios participants lack  information
about other participants, and specifically about the preferences of
other participants.  The extreme case of such  partial information
scenario is termed   \emph{ambiguity},\footnote{In decision theory literature, the terms \enquote{ambiguity,} \enquote{pure
ambiguity,} \enquote{complete ignorance,} \enquote{uncertainty}
(as opposed to \enquote{risk}), and \enquote{Knightian uncertainty}
are used interchangeably to describe this case of unknown probabilities.} and in our case ambiguity  about the preferences of other participants.
In these scenarios,  not only does a participant not know the preferences
of other participants, but he  cannot even form a belief on them
(that is, he lacks the knowledge to form a probability distribution
over preferences). Hence, one cannot analyze the scenario using the
Bayesian framework, which is the common practice for analyzing partial-information
scenarios, and new tools are needed.\footnote{Clearly, if  a player has information that he can use to construct
a belief about the others, we expect the player to use it. In this
work, we study the extreme case in which one has no reason to assume
the players hold such a belief.}  Similarly, in the computer science literature, algorithms, agents,
and mechanisms are often analyzed without assuming a distribution
on the input space or on the environment.  In this work, we define
and analyze equilibria under ambiguity about the the game. In particular,
we concentrate on equilibria of \emph{games with type ambiguity},
i.e., games with ambiguity about the other players' preferences, namely,
their \emph{type}. Our equilibrium definition is based on a refinement
of Wald's MiniMax principle, which corresponds to the common practice
in computer science of worst-case analysis. 

In Section~\ref{sec:Model}, we define a general model of  games
with ambiguity, similar to Harsanyi's model of games with incomplete
information~\cite{Harsanyi1967},\footnote{As described, for example, in \cite[Def.~10.37 p.~407]{Maschler2013}.}
and derive from it the special case of \emph{games with type ambiguity}.
In this model, the knowledge of player~$i$ on player~$j$ is represented
by a set of types $\TypeSet$. Player~$i$ knows that the type of
player~$j$ belongs to $\TypeSet$, but has no prior distribution
on this set, and no information that can be used to construct one.
Our model also enables us to apply the extensive literature on knowledge,
knowledge operators, and knowledge hierarchy  to ambiguity scenarios. 

Next, we present a novel model for decision making under ambiguity:
 $\LEX$ preferences.  We characterize $\LEX$ in the general framework
of decision making under partial information, and show $\LEX$ is
the unique  finest preference that satisfies a few natural properties.
Specifically, we claim that these properties are  satisfied by rational
players in games with type ambiguity, and hence that $\LEX$ is the
right tool of analysis.

Finally,  we derive the respective equilibrium concept, dubbed $\LEXNE$,
and present some of its properties, both in the context of the general
model of games with type ambiguity and in two common  economic scenarios.

\subsubsection*{Wald's MiniMax principle}

A common model for decision making under ambiguity is the MiniMax
principle presented by Wald~\cite{Wald1950}, which we refer to as
$\MIN$ preferences (as distinct from  the $\LEX$ preferences that
we  present later).\footnote{Wald's principle measures actions by their losses rather than by
gains like we do here. Hence, Wald dubbed this principle, which aims
to maximize the worst (minimal) gain, the MiniMax principle while
we dub it $\MIN$.} In the $\MIN$ model, similarly to worst-case analysis in computer
science, the preference of the decision maker over actions is based
solely on the set of possible outcomes. E.g., in  games with type
ambiguity,  the possible outcomes are the consequence of playing
the game with the possible types of the other players. An action $a$
is  preferred to another action $b$ if the worst possible outcome
(for the decision maker) of taking action $a$ is  better than the
worst outcome of taking action $b$. This generalizes the classic
preference maximization model: if there is no ambiguity, there is
a unique outcome for each action, and the $\MIN$ decision model coincides
with preference maximization.   The $\MIN$ model  has been used
 for analyzing expected behavior in scenarios of decision making
under ambiguity. For example, ambiguity  about parameters of the environment,
such as  the distribution of prizes (the multi-prior model)~\cite{Gilboa1989}
and ambiguity of the decision maker about his own utility~\cite{Cerreia-Vioglio2015}.
The $\MIN$ model was also applied to define higher-order goals for
a DM, like regret minimization~\cite[Ch.~9]{savage1972foundations2ndEd}
which is applying $\MIN$ when the utility of the DM is the regret
comparing to other possible actions.  In addition, $\MIN$  has
been used for analyzing interactive scenarios with ambiguity, e.g.,
first-price auctions under ambiguity both of the bidders and of the
seller about the ex-ante distribution of  bidders' values~\cite{Bose2006},
and for designing mechanisms assuming ambiguity of the players about
the ex-ante distribution of the other players' private information~\cite{Wolitzky2015},
or assuming they decide according the regret minimization model~\cite{conf/uai/HyafilB04}.

As we show shortly,  the $\MIN$ decision model is too coarse and
offers too little predictive value in some scenarios involving ambiguity
about the other players' types. We show a natural scenario (a small
perturbation of the Battle of the Sexes game~\cite[Ch.~5 Sec.~3]{Luce1957})
in which almost all action profiles are Nash equilibria according
to $\MIN$. Hence, we are looking for a refinement of the $\MIN$
model that breaks indifference in some reasonable way in cases in
which two actions result in equivalent worst outcomes. In Section~\ref{sec:-axiomatizations},
we show that  na\"{\i}vely breaking indifference by applying $\MIN$
recursively\footnote{I.e., when the decision maker faces two actions that have equivalent
worst outcomes, he decides according to the second-worst outcome.} does not suit scenarios with ambiguity about the other players' types
either. We present two game scenarios, and we claim that they are
equivalent in a very strong sense: a player cannot distinguish between
these two scenarios, even if he has enough information to know the
outcomes of all of his actions. Hence, we claim that a rational player
should act the same way in these two scenarios. Yet, we show that
if a player follows the recursive MIN decision model he plays differently
in the two scenarios. In general, we claim that a decision model for
scenarios with type ambiguity should not be susceptible to this problem,
i.e., it should instruct the player to act the same way in scenarios
if the player cannot distinguish between them. Otherwise, when a decision
rule depends on information which is not visible to the DM, we find
it to be ill-defined.\footnote{This assumption is with accordance to our assumption of ambiguity
which in particular assumes there is no information on the world except
the information on outcomes.} 

\subsubsection*{The $\protect\LEX$ decision model}

In this work we suggest a refinement of Wald's MiniMax principle that
is not susceptible to the above-mentioned problems, and which we term
$\LEX$. According to $\LEX$ the decision maker (DM) picks an action
having an optimal worst outcome (just like under $\MIN$), and breaks
indifference according to the best outcome. We characterize $\LEX$
as the unique finest refinement of $\MIN$ that satisfies three desired
properties (Section~\ref{sec:-axiomatizations}):  monotonicity
in the outcomes, state symmetry, and independence of irrelevant information.
\emph{Monotonicity in the outcomes} is a natural rationality assumption
stating that the DM (weakly) prefers an action $a$ to an action $b$
if in every state of the world (in our framework, a state is a vector
of types of the other players), action $a$ results in an outcome
that is at least as good as the outcome of action $b$ in this state.
\emph{State symmetry} asserts that the decision should not depend
on the names of the states and should not change if the names are
permuted.\emph{ Independence of irrelevant information} asserts that
the DM should not be susceptible to the irrelevant information bias,
describes above. That is, the DM's decision should depend only on
state information that is relevant to his utility.  Specifically,
it requires that if two states of the world have the same outcomes
for each of the actions,  the distinction between the two should
be irrelevant for the DM, and his preference over actions should not
change in case he considers these two states as a single state. We
show that these properties characterize the family of preferences
that are determined by only the worst and the best outcomes of the
actions. \global\long\def\ArbRule{{\rm P}}
 Moreover, we show  that $\LEX$ is the finest refinement of $\MIN$
 in this family: for any preference $\ArbRule$ that satisfies the
three properties,  if $\ArbRule$ is a refinement of $\MIN$,\footnote{That is, for any two actions $a$ and $b$, if $a$ is strongly preferred
to $b$ according to $\MIN$, then $a$ is strongly preferred to $b$
according to $\ArbRule$ too.} then $\LEX$ is a refinement of $\ArbRule$.

\subsubsection*{Equilibrium under $\protect\LEX$ preferences}

In Section~\ref{sec:Model}, we define $\MINNE$ to be the Nash equilibria
under $\MIN$ preferences, that is, the set of action profiles in
which each player best-responds to the actions of other players, and
similarly we define $\LEXNE$ to be the Nash equilibria under $\LEX$
preferences. 

We show that for every game with ambiguity, a $\MINNE$ in mixed strategies
always exists (Thm.~\ref{thm:NEexists}). On the other hand, we show
that there are generic games with ambiguity in which the set of $\MINNE$
is unrealistic and too large to be useful. This holds even for cases
in which the ambiguity is symmetric (all players have the same partial
knowledge) and is  only about the other players' preferences. Here
once again is our motivation for studying  the equilibria under $\LEX.$
On the other hand, we present a simple generic two-player game with
type ambiguity for which no $\LEXNE$ exists. We note that since $\LEX$
is the unique finest refinement of $\MIN$ (which satisfies some properties),
the equilibria of a game with ambiguity under any other refinement
of $\MIN$  is a super-set of the set of $\LEXNE$. Hence, one can
think of $\LEXNE$ as the set of equilibria that do not depend on
assumptions on the tie-breaking rule over $\MIN$ applied by the players.

We show that the problem of finding a $\MINNE$ is a \ProblemClassComplete{PPAD}
problem~\cite{Papadimitriou1994,Papadimitriou2015}, just as finding
a Nash equilibrium when there is no ambiguity.\footnote{When there is no ambiguity, a Nash equilibrium is also a $\MINNE$
of the game and vice-versa. Hence, finding a $\MINNE$ is a \ProblemClassHard{PPAD}
problem. We show it belongs to \ProblemClass{PPAD} and so adding
ambiguity does not make the problem harder.}

\subsubsection*{Applications of $\protect\LEXNE$ to economic scenarios }

 To understand the benefits of analysis using the $\LEX$ model,
we apply the equilibrium concept $\LEXNE$ to two well studied economic
scenarios \textendash{} coordination games and bilateral trade games
\textendash{} while introducing ambiguity. We show that in both scenarios,
a $\LEXNE$ in pure strategies always exists, and we analyze these
equilibria. 

\subsubsection*{Coordination games}

In coordination games, the players simultaneously choose a location
for a common meeting. All players prefer to choose a location that
maximizes the number of players they meet, but they differ in their
tie-breaking rule, i.e.,  their preference over the possible locations.
This is a generalization of the game battle of the sexes~\cite[Ch.~5, Sec.~3]{Luce1957},
and it models economic scenarios where the players need to coordinate
a common action, like agreeing on a meeting place, choosing a technology
(e.g., cellular company), and locating a public good when the cost
is shared, as well as Schelling's focal point experiments~\cite[pp.~54--57]{Schelling1980}:
the parachuters' problem  and the meeting in NYC problem. 

\global\long\def\LL{LL}
 \global\long\def\L{L}
 \global\long\def\R{R}
 \global\long\def\RR{RR}

Take for example, a linear street with four possible meeting locations:
$\LL$, $\L$, $\R$, and $\RR$ (see figure for the distances between
the locations), and two players who choose locations simultaneously
in an attempt  to meet each other.

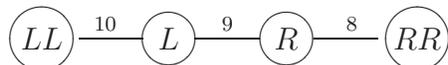
\[
\xymatrix{*++[o][F]{\LL}\ar@{-}^{10}[r] & *++[o][F]{\L}\ar@{-}^{9}[r] & *++[o][F]{\R}\ar@{-}^{8}[r] & *++[o][F]{\RR}}
\]
Both players prefer the meeting  to take place, but have different
preferences over the meeting places (and both are indifferent about
their action if the meeting does not take place). We also assume that
the preference of each player is determined by the distance between
the meeting place and his initial position.\footnote{For example, the preference of a player whose position is $\L$ is
$\L\SprefOver\R\SprefOver\LL\SprefOver\RR$.

For clarity, in this example we use \emph{position} for the initial
position of a player determining his preference, and \emph{location}
for the actions and outcomes.} 

Next, consider a scenario in which each player might be positioned
in any one of the four locations, but he does not know the position
of the other player; i.e., both have ambiguity about the type of the
other player (and hence about his preference). 

First, we notice that profiles in which both players choose, regardless
of their types, the same location are $\LEXNE$ of the game.  From
the perspective of Player~2, if all the types of Player~1 choose
the same location, then choosing this location and meeting Player~1
for sure is strictly preferred (regardless of the position of Player~2)
to any other choice, which would surely result in no meeting. The
analysis for Player~1 is identical, hence these profiles are $\LEXNE$.

Next, consider a case in which Player~1 goes either to $\LL$ or
to $\RR$ (i.e., the  locations chosen by the types of Player~1
are these two locations). From Player~$2$'s perspective, all of
his actions are equivalent in terms of their worst outcome, as there
is always a possibility of not meeting Player~1 (when facing a type
of Player~$1$ who chose a different location). Thus, Player~2 chooses
according to the best outcome for him. Only the actions $\LL$ and
$\RR$ result in a possibility to meet Player~1, i.e., to meet one
of the types of Player~1. Hence Player~2 strictly prefers $\LL$
and $\RR$ to $\L$ and $\R$,  and so he will choose between  $\LL$
and $\RR$ according to his preference over them. Following this reasoning,
we show that the profile in which each type of each of the players
goes to the location closest to him out of $\LL$ and $\RR$ is 
$\LEXNE$, and that any other profile in which one of the players
plays $\LL$ and $\RR$ (i.e., the locations played by his types are
these two locations) is not  $\LEXNE$.

This simple scenario also demonstrates the drawback of using the $\MIN$
model to analyze games with type ambiguity. When Player~1 goes to
either $\LL$ or $\RR$, Player~2 is indifferent between the worst
outcomes of all of his actions, and so, according to the $\MIN$ model,
Player~2 is indifferent between all of his actions. Particularly,
Player~$2$ of type $\LL$ is indifferent between playing $\LL$
(which is his own position), playing $\RR$ (the location farthest
from him), and playing $\L$ (in which he is certain not to meet Player~$1$).
This seems highly unrealistic: we would expect a rational Player~$2$
of type $\LL$ to prefer playing $\LL$ to $\RR$ or $\L$. When using
the $\MIN$ model to analyze equilibria we get that almost all profiles
are $\MINNE$\footnote{In general, for a coordination game over $m$ locations and $n$ players
s.t. each of them has at least $t$ types, more than $\left(1-\sfrac{1}{m^{t-1}}\right)^{n}$-fraction
of the pure action profiles are $\MINNE$ of the game.} including, for instance, the profile in which each type goes to the
location \textbf{farthest} from him among $\LL$ and $\RR$.

We show that in general, for every coordination game with type ambiguity,
the pure Nash equilibria of the no-ambiguity case in which all players
choose the same unique location are also $\LEXNE$.\footnote{In general, if an action profile $a$ is a Nash equilibrium for all
the possible combinations of types, then $a$ is also a $\LEXNE$,
since we expect the players to follow the action profile when the
information about the types is irrelevant to their decision.} Note that the set of equilibria when there is no ambiguity does not
depend on the players' types (which are only a tie-breaking rule between
two locations having the same number of other players). Yet, we get
that ambiguity about the types gives rise to new equilibria, which
we characterize in Section~\ref{sec:Coordination-Games}.  We show
that an equilibrium is uniquely defined by a set of meeting locations
($\LL$ and $\RR$ in the example above) to be the action profile
in which each type of each player chooses his optimal location in
the set. Finally, we also characterize the equilibria for several
cases in which we assume a natural homogeneity constraint on the type
sets. The constraint we choose, taking the type sets to be single-peaked
consistent w.r.t. a line, restricts the ambiguity about the other
players' preferences in a natural way, and hence its impact on the
set of equilibria  is informative for the study of equilibria under
type ambiguity.

\subsubsection*{Bilateral trade}

The second scenario we analyze is bilateral trade. These are two-player
games between a seller owning an item and a buyer who would like to
purchase the item. Both players are characterized by the value they
attribute to the item (their   respective willingness to accept and
willingness to pay). In the mechanism that we analyze, both players
simultaneously announce a price and if the price announced by the
buyer is higher than the one announced by the seller, then a transaction
takes place and the price is the average of the two.\footnote{Our result also holds for a more general case than setting the price
to be the average.} For simplicity, we  assume that a player has the option not to participate
in the trade.\footnote{This option is equivalent to the option of the seller declaring an
extremely high price that will not be matched (and similarly for the
buyer).}

When there is no ambiguity, an equilibrium that includes a transaction
consists of a single  price, which is announced by both players.
When there is ambiguity about the values, we show that, in addition
to the single-price equilibrium, a new kind of equilibrium emerges.
For instance, consider the case in which the value of the buyer can
be any value between 20 and 40 and the value of the seller can be
any value between 10 and 30. First, we notice that there are equilibria
that are based on one price as above,\footnote{An equilibrium in which both players choose (as a function of their
value) either to announce a price common to both or not to participate. } but in any such equilibrium there will be types (either of the buyer
or of the seller) that will prefer not to participate.  If, for example,
the price is 25 or higher, then there are types of the buyer that
value the item at less than this price and will prefer not to participate;
similarly, for prices below 25 there are possible sellers who value
the item at more than 25 and hence will prefer not to participate.
We can show further a $\LEXNE$ with two prices, 15 and 35, in which
both players participate regardless of their type: the seller announces
35 if his value is higher than 15 and 15 otherwise,  and the buyer
announces 35 if his value is higher than 35 and 15 otherwise.\footnote{I.e., the seller announces the lower of the two if both are acceptable
to him, and the higher otherwise; and the buyer announces the higher
of the two if both are acceptable to him, and the lower otherwise.} In this profile, a buyer who values the item at more than 35 prefers
buying the item at either price to not buying it, and hence he best-responds
by announcing 35 and buying the item for sure. A buyer who values
the item at less than 35 prefers buying the item at 15 to not buying
it, and not buying the item to buying it at 35. The best worst-case
outcome he can guarantee is not buying the item (e.g., by announcing
any value between 15 and his value). Based on the worst outcome the
buyer is indifferent between these announcements. Hence, in choosing
between these announcements according to the best-outcome (i.e., meeting
a seller who announces 15), he best-responds by announcing 15. A similar
analysis shows that also the seller best-responds in this profile. 

 We characterize the set of $\LEXNE$ for bilateral trade games,
and in particular we show that for every bilateral trade game, an
equilibrium consists of at most two prices. As a corollary we characterize
the cases for which there exists a full-participation $\LEXNE$, i.e.,
equilibria in which both players choose to participate regardless
of their value (but their bid in these equilibria might depend on
the value).

\section{Model}

\label{sec:Model}

We derive our model for games with type ambiguity as a special
case of a more general model of games with ambiguity. A \emph{game
with ambiguity}\footnote{For simplicity we define a finite game but the definitions extend
to infinite cases as well. Our definitions of preferences and results
also extend under minor technical assumptions.} is a vector
\[
\left\langle \Players,\left(\ActionSet^{i}\right)_{i\in\Players},\Omega,\left(u^{i}\right)_{i\in\Players},\left(\TypeSet^{i}\right)_{i\in\Players}\right\rangle 
\]
where:
\begin{itemize}
\item $\left\langle \Players,\left(\ActionSet^{i}\right)_{i\in\Players}\right\rangle $
is an $n$-player game form. That is,

\begin{itemize}
\item[] $\Players$ is a finite set of \emph{players} $N=\left\{ 1,\ldots,n\right\} $;
\item[] $\ActionSet^{i}$ is a finite set of \emph{actions} of player~$i$,
and we denote by $\ActionSet$ the set of \emph{action profiles} $\times_{i\in\Players}\ActionSet^{i}$.
\end{itemize}
\item $\Omega$ is a finite set of \emph{states of the world}.
\item $u^{i}\colon\Omega\times\ActionSet\rightarrow\Re$ is a \emph{utility
function} for player~$i$ that specifies his utility from every state
of the world and profile of actions. We identify $u^{i}$ with its
linear extension to mixed actions, $u^{i}\colon\Omega\times\Delta\left(\ActionSet^{i}\right)\rightarrow\Re$,
where $\Delta\left(\ActionSet^{i}\right)$ is the set of mixed actions
over $\ActionSet^{i}$.\footnote{An implicit assumption here is that the players hold vNM preferences,
that is,  they evaluate a mixed action profile by its expectation.
This does not restrict the modeling of preferences under ambiguity.
Using the terminology of Anscombe and Aumann~\cite{Anscombe1963},
we distinguish between roulettes and horse races.}
\item $\left\langle \Players,\Omega,\left(\TypeSet^{i}\right)_{i\in\Players}\right\rangle $
is an Aumann model of incomplete information. That is, $\TypeSet^{i}$
is a partition of $\Omega$ to a finite number of partition elements
($\Omega=\dot{\cup}_{t^{i}\in\TypeSet^{i}}t^{i}$).  We refer to
$t^{i}\in\TypeSet^{i}$ as a \emph{type} of player~$i$.\footnote{For a full definition of Aumann's model and its descriptive power,
see, e.g., \cite{Aumann1976,Aumann1999} and \cite[Def.~9.4 p.~323]{Maschler2013}.
As described in \cite{Aumann1999} this model is equivalent to defining
$\TypeSet^{i}$ using signal functions and to defining them using
knowledge operators (i.e., the systematic approach).}
\end{itemize}
The above is commonly known by the players. A game proceeds as follows.
\begin{itemize}
\item Nature chooses (arbitrarily) a state of the world $\omega\in\Omega$.
\item Each player is informed (only) about his own partition element $t^{i}\in\TypeSet^{i}$
satisfying $\omega\in t^{i}$. 
\item The players play their actions simultaneously: Player~$i$, knowing
his type $t^{i}$, selects a (mixed) action $a^{i}\in\Delta\left(\ActionSet^{i}\right)$.
\item Every player gets a payoff according to $u$: Player~$i$ gets $u^{i}\left(t,a\right)$,
where $a=\left(a^{1},a^{2},\ldots,a^{2}\right)$ is the action profile
and $t=\left(t^{1},t^{2},\ldots,t^{2}\right)$ is the type profile.
\end{itemize}
Notice that the difference between this model and the standard model
of games with incomplete information~\cite{Harsanyi1967} (e.g.,
as described in~\cite[Def.~10.37 p.~407]{Maschler2013}) is that
in the latter it is assumed that the players also  have posterior
distributions on $t^{i}$ (or equivalently, they have subjective prior
distributions on $\Omega$).

In this work we are interested in \emph{games with type ambiguity}.
In these games the states of the world are types vectors $\Omega\subseteq\times_{i=1}^{n}\TypeSet^{i}$,
i.e., the unknown information can be represented as information on
the types, and in particular any two states of the world are distinguishable
by at least one player.\footnote{For Bayesian settings, this assumption is without loss of generality,
because we can unify two indistinguishable states and replace them
by the respective lottery, without changing the preferences. Since
here we assume no posterior distribution, this assumption is indeed
constraining.} For this restricted model, we justify our choice of $\LEX$ preferences.
Note that we prove the existence of a mixed $\MINNE$ (Thm.~\ref{thm:NEexists})
for every game with ambiguity.

A strategy of a player states his action for each of his types $\sigma^{i}\colon\TypeSet^{i}\rightarrow\Delta\left(\ActionSet^{i}\right)$.
 Given a type profile $t=\left(t^{1},\ldots,t^{n}\right)$ and a
strategy profile $\sigma=\left(\sigma^{1},\ldots,\sigma^{n}\right)$,
we denote by $t^{-i}$ the types of the players besides player~$i$
and by $\sigma^{-i}\left(t^{-i}\right)$ their actions under $t$
and $\sigma$. I.e., $\sigma^{-i}\left(t^{-i}\right)=\left(\sigma^{1}\left(t^{1}\right),\ldots,\sigma^{i-1}\left(t^{i-1}\right),\sigma^{i+1}\left(t^{i+1}\right),\ldots,\sigma^{n}\left(t^{n}\right)\right)$.
We note that the utility of a given type of player~$i$ is only
affected by that actions taken by other players and not by the actions
of the other types player~$i$. Hence, we assume a player chooses
his action after knowing his type and not ex-ante beforehand, and
model player~$i$'s choice of action (best-responding to the others)
as a series of independent problems, one for each of his types, of
choosing an action. We refer to these problems as the decision process
carried out by a type.

\subsection{Preferences under ambiguity}

Decision theory~(\cite[Ch.~13]{Luce1957}, \cite{Gilboa2011}) 
deals with scenarios in which a single decision maker (DM) needs to
choose an action from a given set $\ActionSet$ when his utility from
an action $a\in\ActionSet$ depends also on an unknown state of the
world $\omega\in\Omega$, and so his preference is represented by
a utility function $u\colon\ActionSet\times\Omega\rightarrow\Re$.
Player~$i$ (of type $t^{i}$) looks for a response (an action) to
a profile $\sigma^{-i}$. This response problem is of the same format
as the DM problem: he needs to choose an action while not knowing
the state of the world $\omega$ (the types of his opponents $t^{-i}$
and their actions $\sigma^{-i}\left(t^{-i}\right)$ are derived from
$\omega$). 

We define the two preference orders over actions, $\MIN$ and $\LEX$,
in the framework of Decision Theory.  We define them by defining
the pair-wise comparison relation, and it is easy to see that this
relation is indeed an order. The first preference we define corresponds
to Wald's MiniMax decision rule~\cite{Wald1950}.

\begin{defn}[$\MIN$ preference]
\MyLyxThmNewline{} A DM strongly prefers an action $a$ to an action
$a'$ according to $\MIN$, if the worst outcome when playing $a$
is preferred to the worst outcome of playing $a'$:\footnote{The $\MIN$ preference is representable by a utility function $U\left(a\right)=\min_{\omega\in\Omega}u\left(a,\omega\right)$. }\VspaceBeforeMath{-.7em}{}
\[
\min_{\omega\in\Omega}u\left(a,\omega\right)>\min_{\omega\in\Omega}u\left(a',\omega\right).
\]
\end{defn}
These preferences follow the same motivation as worst-case analysis
of computer science (where a designer needs to choose an algorithm
or a system to use and the expected environment in unknown in advance\footnote{Note that also the commonly used competitive-ratio is actually applying
the $\MIN$ preference when we define $u\left(a,\omega\right)$ not
to be the performance of an algorithm $a$ on an input $\omega$,
but the ratio between this performance and the performance of an optimal
all-knowing algorithm.}). $\MIN$ preference can also be justified as an extreme ambiguity
aversion; judging an action by the worst possible outcome ignoring
the probability of this outcome. In Section~\ref{sec:-axiomatizations}
(Prop.~\ref{prop:GS+mine=00003DMINexp}) we show an axiomatization
of $\MIN$ preference as the unique preference in the ambiguity aversion
maxmin model presented by Gilboa and Schmeidler ~\cite{Gilboa1989}
(i.e., that satisfies the axioms they present), which satisfies two
additional axioms we present here.

The second preference we introduce is a refinement of the $\MIN$
preference, as it breaks ties in cases where $\MIN$ states indifference
between actions.
\begin{defn}[$\LEX$ preference]
\MyLyxThmNewline{} A DM strongly prefers an action $a$ to an action
$a'$ according to $\LEX$, if either $\min_{\omega\in\Omega}u\left(a,\omega\right)>\min_{\omega\in\Omega}u\left(a',\omega\right)$
or he is indifferent between the two respective worst outcomes and
he prefers the best outcome of playing $a$ to the best outcome of
playing $a'$:\footnote{The $\LEX$ preference is not representable by a utility function,
for the same reason that the lexicographic preference over $\Re^{2}$
is not representable by a utility function~\cite[Ch.~3.C, p.~46]{Mas-Colell1995}.}\VspaceBeforeMath{-.7em}{}
\[
\left\{ \begin{array}{l}
\min\limits _{\omega\in\Omega}u\left(a,\omega\right)=\min\limits _{\omega\in\Omega}u\left(a',\omega\right)\\
\max\limits _{\omega\in\Omega}u\left(a,\omega\right)>\max\limits _{\omega\in\Omega}u\left(a',\omega\right).
\end{array}\right.
\]
\end{defn}

Returning to our framework of games with ambiguity, we define the
corresponding \emph{best response} ($\BR$) correspondences: $\MINBR$
and $\LEXBR$. The best response of a (type of a) player is a function
that maps any action profile of the other players to the  actions
that are optimal according to the preference. It is easy to see that
a best response according to $\LEX$ is also a best response according
to $\MIN$; that is, $\LEXBR$ is a refinement of $\MINBR$. We show
that the two best response notions are well defined and exist for
any (finite) game.\footnote{This lemma can be easily extended to the case of an infinite number
of states by assuming some structure on the action set and the utility
function.}
\begin{lem}
\MyLyxThmFN{Proof can be found in Appendix~\ref{Appendix:Proof of lem:BRexists}.}\label{lem:BRexists}\MyLyxThmNewline{}

The following best response correspondences are non-empty: pure $\MINBR$,
mixed $\MINBR$, pure $\LEXBR$, and mixed $\LEXBR$.

\end{lem}

\subsection{Equilibria under ambiguity}

Next we define the corresponding (interim) Nash equilibrium ($\NE$)
concepts as the profiles of strategies in which each type best-responds
to the strategies of the other players. From the definition of $\LEX$
it is clear that any equilibrium according to $\LEX$ is also an equilibrium
according to $\MIN$. Hence we regard $\LEXNE$ as an equilibrium-selection
notion or a refinement of $\MINNE$, in cases in which we find $\MINNE$
to be unreasonable. Our main theorem for this section is showing that
any game with ambiguity has an equilibrium according to $\MIN$ ($\MINNE$)
in mixed strategies.\footnote{Throughout this paper, unless stated otherwise, when we refer to $\MINNE$
and $\LEXNE$ we mean equilibria in pure actions.} 

\begin{thm}
\label{thm:NEexists}\MyLyxThmNewline{}
Every game with ambiguity has a $\MINNE$ in mixed actions.

\end{thm}
\begin{proof-sketch}
(The full details of the proof can be found in Appendix~\ref{Appendix:Proof of thm:NEexists})\global\long\def\SET{\mathbb{S}}
\MyLyxThmNewline{}

We take $\SET$ to be the set of all profiles of mixed strategies
of the types and define the following set-valued function $F:\mathbb{\SET}\rightarrow\mathbb{\SET}$.
Given a strategy profile $s$, $F\left(s\right)$ is the product of
the best responses to $s$ (according to $\MIN$) of the different
types. We prove the existence of a mixed $\MINNE$ by applying Kakutani's
fixed point theorem~\cite{Kakutani1941} to $F$. A fixed point of
$F$ is a profile $s$ satisfying $s\in F\left(s\right)$; i.e., each
type best-responds to the others in the profile $s$, and hence $s$
is a $\MINNE$.
\end{proof-sketch}
Since the existence of $\MINNE$ is the result of applying Kakutani's
fixed point theorem to the best response function,\footnote{The best response correspondence can be computed in polynomial time.\VspaceBeforeMath{-.7em}{}
\[
\begin{array}[b]{rl}
\BR\left(s\right) & =\argmax_{\sigma^{i}}\min_{\omega\in\Omega}u\left(\omega,\sigma^{i},s^{-i}\left(t^{-i}\left(\omega\right)\right)\right)\\
 & =\argmax_{\sigma^{i}}\min_{\omega\in\Omega}\Exp_{a\sim\sigma^{i}}C_{a,\omega}
\end{array}\text{ for }C_{a,\omega}=u\left(\omega,a,s^{-i}\left(t^{-i}\left(\omega\right)\right)\right).
\]
The maximal value a player can guarantee himself, $v^{*}=\max_{\sigma^{i}}\min_{\omega\in\Omega}\Exp_{a\sim\sigma^{i}}C_{a,\omega}$,
is the solution to the following program which is linear in $v$ and
$\sigma^{i}$\VspaceBeforeMath{-.7em}{}
\[
\max v\st\forall w\,\Exp_{a\sim\sigma^{i}}C_{a,\omega}\geqslant v,
\]
that can be solved  in polynomial time. Given $v^{\star}$, $\BR\left(s\right)$
is the intersection of $\sizeof{\Omega}$ hyperplanes of the form
$\Exp_{a\sim\sigma^{i}}C_{a,\omega}\geqslant v^{*}$.} we get as a corollary the complexity of the problem of finding $\MINNE$. 
\begin{cor}
\MyLyxThmNewline{}

The problem of finding a $\MINNE$ is in \ProblemClass{PPAD}~\cite{Papadimitriou1994,Papadimitriou2015}.
Moreover, it is a \ProblemClassComplete{PPAD} problem since a special
case of it, namely, finding a Nash equilibrium, is a \ProblemClassHard{PPAD}
problem~\cite{Chen2009}.
\end{cor}
Next we show that there are games that  with no equilibrium according
to $\LEX$. We show that this is true even for a simple generic game:
a two-player game with \emph{type ambiguity} on \emph{one side only}. 

\begin{lem}
\MyLyxThmNewline{}There are games for which there is no $\LEXNE$.
\end{lem}
\begin{proof}
\MyLyxThmNewline{}
Let $G$ be the following two-player game with two actions for each
of the players. The row player's utility is $\begin{array}{|c|c|c|}
\hline  & L & R\\
\hline T & 0 & 0\\
\hline B & -1 & 1
\\\hline \end{array}\vspace{.5em}$. The column player is one of two types: either having utility $\begin{array}{|c|c|c|}
\hline  & L & R\\
\hline T & 0 & 1\\
\hline B & 0 & -2
\\\hline \end{array}$ or $\begin{array}{|c|c|c|}
\hline  & L & R\\
\hline T & 0 & 2\\
\hline B & 0 & -1
\\\hline \end{array}\vspace{.2em}$ (and the row player does not know which).

Then, in the unique $\MINNE$ the first type of the column player
mixes $\half L+\half R$, the second type of the column player plays
$R$, and the row player mixes $\frac{2}{3}T+\frac{1}{3}B$ (all his
mixed actions give him a worst-case payoff of 0). But this is not
a $\LEXNE$ since the row player prefers to deviate to playing $B$
for the possibility of getting $1$, and hence the game does not have
$\LEXNE$ in mixed strategies.\footnote{Technical comment: The reason Kakutani's theorem cannot be applied
here (besides its result being wrong) is twofold:
\begin{itemize}
\item The best-response set is not convex: 

Consider a player who has two possible pure actions, $T$ and $B$,
and his utility (as a function of the action of the opponent) is $\begin{array}{|c|c|c|c|}
\hline  & L & M & R\\
\hline T & 0 & 1 & 2\\
\hline B & 0 & 2 & 1
\\\hline \end{array}\vspace{.5em}$. Next, consider he faces one of three types of his opponent who play
the three actions, respectively. He is indifferent between his two
actions (both give him $0$ in the worst case and $2$ in the best
case), but strictly prefers the two pure actions to any mixture of
the two (giving him less than $2$ in the best case).
\item The best-response function is not upper semi-continuous:

In the example in the lemma, when the row player faces one type that
plays the pure strategy $R$ and another type that mixes $\left(\half+\epsilon\right)L+\left(\half-\epsilon\right)R$,
his unique best response is to play $T$ for any $\epsilon>0$, but
to play $B$ for $\epsilon=0$.
\end{itemize}
} 

\end{proof}

\section{Axiomatization of $\protect\LEX$ }

\label{sec:-axiomatizations}

\global\long\def\MIN{\text{{\rm MIN}}}
\global\long\def\MAX{\text{{\rm MAX}}}
\global\long\def\LEX{\text{{\rm \ensuremath{\MIN}then\ensuremath{\MAX}}}}

\global\long\def\MINBR{\text{{\rm \ensuremath{\MIN}-BR}}}
\global\long\def\LEXBR{\text{{\rm \ensuremath{\LEX}-BR}}}
\global\long\def\BR{\text{{\rm BR}}}

\global\long\def\MINNE{\text{{\rm \ensuremath{\MIN}-NE}}}
\global\long\def\LEXNE{\text{{\rm \ensuremath{\LEX}-NE}}}
\global\long\def\NE{\text{{\rm NE}}}

\global\long\def\ActionSet{\mathcal{A}}

\global\long\def\TypeSet{\mathcal{T}}

\renewcommand\MINBR{\texify{\MIN-BR}\xspace}
\renewcommand\MINNE{\texify{\MIN-NE}\xspace}
\renewcommand\LEXBR{\texify{\LEX-BR}\xspace}
\renewcommand\LEXNE{\texify{\LEX-NE}\xspace}

\renewcommand\LEX{\texify{\texttt{MIN\-then\-MAX}}\xspace}
\renewcommand\LEX{\texify{\scshape{MIN\-then\-MAX}}\xspace}
\renewcommand\MIN{\texify{\scshape{MIN}}\xspace}



In this section we justify using equilibria under $\LEX$ preferences
for the analysis of games with type ambiguity. To do so, we present
three properties for decision making under ambiguity and characterize
$\LEX$ as the finest refinement of $\MIN$ that satisfies them (Thm.~\ref{thm:LEX characterization}).
We claim that these properties are necessary for modeling decision
making under ambiguity about the other players' types. In doing so,
we justify our application of $\LEXNE$.

\subsection{The decision-theoretic framework}

Let $\Omega$ be a finite set of states of the world. We characterize
a preference, i.e., a total order, of a decision maker (DM) over the
action set $\ActionSet$ where an action is a function $a\colon\Omega\rightarrow\Re$
that yields a utility for each state of the world.\footnote{Replacing $\Re$ with any other ordered set would not change our results.
}

Our first two properties are natural and we claim that any reasonable
preference under ambiguity should satisfy them. The first property
we present is a basic rationality assumption: monotonicity. It requires
that if an action $a$ results in a higher or equal utility than an
action $b$ in all states of the world, then the DM should weakly
prefer $a$ to $b$.
\begin{ax}[Monotonicity]
\MyLyxThmNewline{} For any two actions $a$ and b, if $a\left(\omega\right)\geqslant b\left(\omega\right)$
for all $\omega\in\Omega$, then either the DM is indifferent between
the two or he prefers $a$ to $b$.
\end{ax}
The second property, state symmetry, states that the DM should choose
between actions based on properties of the actions and not of the
states. I.e., if we permute the states' names, his preference should
not change. Since we can assume that the states themselves have no
intrinsic utility beyond the definition of the actions, this property
formalizes the property that the DM, sue to the ambiguity about the
state, should satisfy the \emph{Principle of Insufficient Reason}
and treat the states symmetrically.\footnote{Note that this property rules out any subjective expectation maximization
preference, except for expectation under the uniform distribution.}
\begin{ax}[State symmetry]
\MyLyxThmNewline{} For any two actions $a$ and $b$ and a bijection
$\psi\colon\Omega\rightarrow\Omega$, if $a$ is preferred to $b$,
then $a\circ\psi$ is preferred to $b\circ\psi$ ($a\circ\psi\left(\omega\right)$
is defined to be $a\left(\psi\left(\omega\right)\right)$, i.e., the
outcome of the action $a$ in the state $\psi\left(\omega\right)$).
\end{ax}
The last property we present is independence of irrelevant information.
This property requires that if the DM considers one of the states
of the world as being two states, by way of considering some new parameter,
 his preference should not change. We  illustrate the desirability
of this property for games with type ambiguity using the following
example. Consider the following variant of the Battle of the Sexes
game between Alice and Bob, who need to decide on a joint activity:
either a Bach concert ($B$) or a Stravinsky concert ($S$). Taking
the perspective of Alice, assume that she faces one of two types of
Bob: $Bob^{B}$ whom she expects to choose $B$, or $Bob^{S}$ whom
she expects to choose $S$. Assume that Alice prefers $B$, ans so
her valuation of actions is \vspace{.5em}$\begin{array}{|c|c|c|}
\hline  & Bob^{B} & Bob^{S}\\
\hline B: & 2 & 0\\
\hline S: & 0 & 1
\\\hline \end{array}$ ($0$ if they do not meet and $2$ or $1$ if they jointly go to
a concert). But there might be other information Alice does not know
about Bob. For example, it might be that in case Bob prefers (and
chooses) $S$, Alice also does not know his favorite soccer team.\footnote{Of course, his favorite soccer team is clear in case he prefers Bach.}
So she might actually conceive the situation as \vspace{.5em}$\begin{array}{|c|c|c|c|}
\hline  & Bob^{B} & Bob^{S,\star} & Bob^{S,\dag}\\
\hline B: & 2 & 0 & 0\\
\hline S: & 0 & 1 & 1
\\\hline \end{array}$. Since this new soccer information is irrelevant to the game, it
should not change the action of a rational player. Notice that if
Alice chooses according to the recursive $\MIN$ rule we described
in the introduction, she will choose according to the second-worst
outcome and hence choose $B$ in the first scenario and $S$ in the
second scenario. We find a decision model of a rational player which
is susceptible to this problem  to be an ill-defined model.

\global\long\def\doublehatomega{\WILLNOTWORK}
\newcommand{\doublehat}[1]{\widehat{\vphantom{\rule{1pt}{5.5pt}}\smash{\widehat{#1}}}}
\renewcommand{\doublehatomega}{\doublehat{\omega}}
\begin{ax}[Independence of irrelevant information]
 \MyLyxThmNewline{} Let $a$ and $b$ be two actions on $\Omega$
s.t. $a\prefBy b$, and let $\widehat{\omega}\in\Omega$ be a state
of the world. Define a new state space $\Omega'=\Omega\dot{\cup}\left\{ \doublehatomega\right\} $
and let $a'$ and $b'$ be two actions on $\Omega'$ satisfying $a'\left(\omega\right)=a\left(\omega\right)$
and $b'\left(\omega\right)=b\left(\omega\right)$ for all states $\omega\in\Omega\setminus\left\{ \hat{\omega}\right\} $,
$a'\left(\doublehatomega\right)=a'\left(\hat{\omega}\right)=a\left(\hat{\omega}\right)$
, and $b'\left(\doublehatomega\right)=b'\left(\hat{\omega}\right)=b\left(\hat{\omega}\right)$.
Then $a'\prefBy b'$.
\end{ax}
We show that $\LEX$ is the finest refinement of $\MIN$ that satisfies
the above three axioms.\footnote{Notice that the $\MIN$ preference satisfies these properties.}
\begin{thm}
\label{thm:LEX characterization}\MyLyxThmNewline{}
$\LEX$ is the unique preference that satisfies
\begin{itemize}
\item Monotonicity.
\item State symmetry.
\item Independence of irrelevant information.
\item It is a refinement of $\MIN$.\footnote{That is, for any two actions $a$ and $b$, if $a$ is strongly preferred
to $b$ by a DM holding a $\MIN$ preference, then $a$ is strongly
preferred to $b$ by a DM holding a $\LEX$ preference.}
\item It is the finest preference that satisfies the above three properties.
That is, it is a refinement of any preference that satisfies the above
properties.
\end{itemize}

\end{thm}
We claim that the three axioms are necessary for modeling a rational
decision making under type ambiguity: Monotonicity is a basic rationality
axiom, and the other two capture that the DM does not have any additional
information distinguishing between the states of the world besides
the outcomes of his actions. Following Wald, one could define the
family of all refinements of $\MIN$ that satisfy the axioms, and
analyze the equilibria when all players follow models from this family.
Showing that $\LEX$ is a refinement of any of the preferences in
this family, says that the set of equilibria when all players follow
models from  this family, must include all equilibria  for the case
when the players follow the $\LEX$ model (i.e., all $\LEXNE$). Moreover,
$\LEXNE$ are the only profiles which are equilibria whenever all
players follow models from this family. We interpret this result as
robustness of the $\LEXNE$ notion: these are the equilibria an outside
party can expect (e.g., the self-enforcing contracts he can offer
to the players), while not knowing the exact preferences of the players.

\begin{proof-of-something}{\ref{thm:LEX characterization}}{Theorem}
\MyLyxThmNewline{}

In order to prove this theorem we first prove that any preference
that satisfies the first three properties can be defined using the
worst (minimal) and best (maximal) outcomes of the actions.\footnote{Arrow and Hurwicz~\cite{Arrow1972} showed a similar result for decision
rules. They defined four properties (A\textendash D) and showed that
under these properties the decision rule can be defined using the
worst and best outcomes only. Their properties are of a similar nature
to the properties we present: Property~A in~\cite{Arrow1972} derives
that the decision rule is derived by a preference (\cite[Prop.~1.D.2, p.~13]{Mas-Colell1995}),
and Properties B, C, and D are of a similar flavor to the properties
of state symmetry, independence of irrelevant information, and monotonicity
in the outcomes, respectively. In order to avoid defining the framework
of decision rules, and because our proof is simple and different from
theirs, we prove Lemma~\ref{lem:MIN and MAX only} directly and do
not rely on their result.}
\begin{lem}
\MyLyxThmFN{This lemma is proved in Appendix~\ref{Appendix:Proof of lem:MIN and MAX only}}\label{lem:MIN and MAX only}\MyLyxThmNewline{}

Any preference that satisfies monotonicity, state symmetry, and independence
of irrelevant information can be defined as a function of the worst
and the best outcomes of the actions.

\end{lem}
Using this lemma, we turn to prove the characterization of $\LEX$.

It is easy to verify that $\LEX$ indeed satisfies the first four
properties. We note that uniqueness is an immediate consequence of
uniqueness of a finest refinement and prove that indeed $\LEX$ is
a finest refinement. Let $\ArbRule$ be an arbitrary refinement of
$\MIN$ that satisfies Monotonicity, state symmetry, and Independence
of Irrelevant Information. Let $a$ and $b$ be two actions s.t. a
DM holding a $\ArbRule$ preferences strongly prefers $a$ to $b$,
and we'll show that $a$ is preferred to $b$ also by a DM holding
a $\LEX$ preference. Applying Lemma~\ref{lem:MIN and MAX only},
we get that $\ArbRule$ can be defined using the minimal and maximal
outcomes. We denote the respective minimal and maximal outcomes of
$a$ by $m_{a}$ and $M_{a}$, and of $b$ by $m_{b}$ and $M_{b}$.

Assume for contradiction that a DM holding a $\LEX$ preference weakly
prefers $b$ to $a$. Since $\ArbRule$ is a refinement of $\MIN$,
it cannot be that $m_{b}>m_{a}$, and hence we get that $m_{b}=m_{a}$
and $M_{b}\geqslant M_{a}$. Both $\ArbRule$ and $\LEX$ can be defined
as a function of the minimal and maximal outcome, so with no loss
of generality we can assume that $a$ and $b$ result in the same
minimal outcome $m_{a}$ is some state $\omega_{m}$, result in their
respective maximal outcomes in the same state $\omega_{M}$, and result
in the same intermediate outcomes in all other states. Since $\ArbRule$
is monotone, we get that a DM holding a $\ArbRule$ preference weakly
prefers $b$ to $a$ and by that get a contradiction.
\end{proof-of-something}
As a corollary of Lemma~\ref{lem:MIN and MAX only}, we get in addition
an axiomatization of Wald's $\MIN$ rule as the unique rule that satisfies
both the above three properties and the natural axioms of Gilboa and
Schmeidler~\cite{Gilboa1989}: certainty independence, continuity,
monotonicity, and uncertainty aversion.

\begin{prop}
\MyLyxThmFN{This lemma is proved in Appendix~\ref{Appendix:Proof of prop:GS+mine=00003DMINexp}}\MyLyxThmNewline{}\label{prop:GS+mine=00003DMINexp}

When there are at least three states of the world ($\sizeof{\Omega}\geqslant3$),
$\MIN$ is the unique preference over $\Re^{\Omega}$ (i.e., actions
that return cardinal outcomes) that satisfies
\begin{itemize}
\item Let $a$ and $b$ be two actions, $c\in\Re$, and $\alpha\in\left(0,1\right)$.
Then for the two actions $a'$ and $b'$ defined by $a'\left(\omega\right)=\alpha\cdot a\left(\omega\right)+\left(1-\alpha\right)\cdot c$
and $b'\left(\omega\right)=\alpha\cdot b\left(\omega\right)+\left(1-\alpha\right)\cdot c$
for all $\omega\in\Omega$.:\VspaceBeforeMath{-.7em}{}
\[
a\SprefOver b\iff a'\SprefOver b'.
\]
\item Let $a$, $b$, and $c$ be three actions s.t. $a\SprefOver b\SprefOver c$.
Then,
\begin{itemize}
\item[] there exists a scalar $\alpha\in\left(0,1\right)$ and an action
$f_{\alpha}$ defined by $f_{\alpha}\left(\omega\right)=\alpha\cdot a\left(\omega\right)+\left(1-\alpha\right)\cdot c\left(\omega\right)$
for all $\omega\in\Omega$ s.t. $f_{\alpha}\SprefOver b$,
\item[] and there exists a scalar $\beta\in\left(0,1\right)$ and an action
$f_{\beta}$ defined by $f_{\beta}\left(\omega\right)=\beta\cdot a\left(\omega\right)+\left(1-\beta\right)\cdot c\left(\omega\right)$
for all $\omega\in\Omega$ s.t. $f_{\beta}\SprefBy b$.
\end{itemize}
\item For any two actions $a$ and $b$ s.t. $a\sim b$ (i.e., $a\prefOver b$
and $b\prefOver a$), it holds that $c_{\alpha}\prefOver a$ for any
action $c_{\alpha}$ defined by $c_{\alpha}\left(\omega\right)=\alpha\cdot a\left(\omega\right)+\left(1-\alpha\right)\cdot b\left(\omega\right)$
for all $\omega\in\Omega$, for some $\alpha\in\left[0,1\right]$.
\item Monotonicity.
\item State symmetry.
\item Independence of irrelevant information.
\end{itemize}

\end{prop}

\section{Bilateral trade}

\label{sec:Bilateral-Trade}

In order to demonstrate this new notion of equilibrium, $\LEXNE$,
we apply it to two economic scenarios that have type ambiguity. For
both of them we show a $\LEXNE$ in pure strategies always exists
and analyze these equilibria. 
The first scenario we analyze is bilateral trade games.  Bilateral
trade is one of the most basic economic models, which  captures many
common scenarios. It describes an interaction between two players,
a \emph{seller} and a \emph{buyer}. The seller has in his possession
a single indivisible item that he values at $v_{s}$ (e.g., the cost
of producing the item), and the buyer values the item at $v_{b}$.
We assume that both values are private information, i.e., each player
knows only his own value, and we would like to study the  cases in
which the item changes hands in return for money, i.e., a transaction
occurs.\footnote{Another branch of the literature on bilateral trade studies the process
of bargaining (getting to a successful transaction). Since we would
like to study the impact of ambiguity, we restrict our attention to
the outcome.}  Chatterjee and Samuelson~\cite{Chatterjee1983} presented bilateral
trade as a model for negotiations between two strategic agents, such
as settlement of a claim out of court, union-management negotiations,
and of course a model for negotiation on transaction between two individuals
and a model for trade in financial products. The important feature
the authors note is that an agent, while certain of the potential
value he places on a transaction, has only partial information concerning
its value for the other player.\footnote{For instance, in haggling over the price of a used car, neither buyer
nor seller knows the other's walk-away price.} Bilateral trade is also of a theoretical importance, and moreover
a multi-player generalization of it, \emph{double auction}.,\footnote{In a double auction~\cite{Friedman1993}, there are several sellers
and buyers, and we study mechanisms and interactions matching them
to trading pairs.} These models have been used as a tool to get insights into how to
organize trade between buyers and sellers, as well as to study how
prices in markets are determined. 

In this section we assume that there is ambiguity about the players'
values (their types), and we study trading mechanisms, i.e., procedures
for deciding whether the item  changes hands, and how much the buyer
pays for it. We assume that the players are strategic, and hence a
mechanism should be analyzed according to its expected outcomes in
equilibrium. 

 We concentrate on a family of simple  mechanisms (a generalization
of the bargaining rules of Chatterjee and Samuelson~\cite{Chatterjee1983}):
 the seller and the buyer post simultaneously their respective bids,
$a_{s}$ and $a_{b}$, and if $a_{s}\leqslant a_{b}$ the item is
sold for $x\left(a_{s},a_{b}\right)$, for $x$ being a known  monotone
function satisfying $x\left(a_{s},a_{b}\right)\in\left[a_{s},a_{b}\right]$.
For ease of presentation, we add to the action sets of both players
a \enquote{no participation} action $\bot$, which models the option
of a player not to participate in the mechanism; i.e., there is no
transaction whenever one of the players plays $\bot$. This simplifies
the presentation by grouping together profiles in which a player chooses
extreme bids that would not be matched by the other player.   Hence,
the utilities of a seller of type $v_{s}$ and a buyer of type $v_{b}$
from an action profile $\left(a_{s},a_{b}\right)$ are (w.l.o.g.,
we normalize the utilities of both players to zero in the case where
there is no transaction):

\[
\begin{array}{l}
u_{s}\left(v_{s};a_{s},a_{b}\right)=\begin{cases}
a_{s}\leqslant a_{b} & x\left(a_{s},a_{b}\right)-v_{s}\\
a_{s}>a_{b} & 0\\
a_{s}=\bot\vee a_{t}=\bot & 0
\end{cases}\\
u_{b}\left(v_{b};a_{s},a_{b}\right)=\begin{cases}
a_{s}\leqslant a_{b} & v_{b}-x\left(a_{s},a_{b}\right)\\
a_{s}>a_{b} & 0\\
a_{s}=\bot\vee a_{t}=\bot & 0.
\end{cases}
\end{array}
\]
Under full information (i.e., the values $v_{s}$ and $v_{b}$ are
commonly known), there is essentially only one kind of equilibrium:
the \emph{one-price equilibrium}. If $v_{s}\leqslant v_{b}$, the
equilibria in which there is a transaction are all the profiles $\left(a_{s},a_{b}\right)$
s.t. $a_{s}=a_{b}\in\left[v_{s},v_{b}\right]$ (i.e., the players
agree on a price), and the equilibria in which there is no transaction
are all profiles in which both players choose not to participate,
regardless of their type. 

Introducing type ambiguity, we define the seller type set $V_{s}$
and the buyer type set $V_{b}$, where each set holds the possible
valuations of the player for the item. We show that under type ambiguity,
there are at most three kinds of equilibria, and we fully characterize
the equilibria set. We show that in addition to the above no-transaction
equilibria and one-price equilibria, we get a new kind of equilibrium:
the \emph{two-price equilibri}um. In such an equilibrium, both the
seller and the buyer participate regardless of their valuations, and
bid one of two possible prices: $p_{L}$ and $p_{H}$. For some type
sets, namely $V_{s}$ and $V_{b}$, these two prices are the only
full-participation equilibria, i.e., equilibria in which both players
choose to announce a price and participate, regardless of their value. 
\begin{lem}
\MyLyxThmFN{This result is also valid, and even more natural, for infinite type
sets.}\label{lem:Bilateral Trade - Main Result}\MyLyxThmNewline{}
\global\long\def\ls{\min V_{s}}
\global\long\def\hs{\max V_{s}}
\global\long\def\lb{\min V_{b}}
\global\long\def\hb{\max V_{b}}

Let $G$ be a bilateral trade game defined by a price function $x\left(a_{s},a_{b}\right)$
and two type sets $V_{s}$ and $V_{t}$, both having a minimum and
maximum.\footnote{We state the result here for the case where both sets have a minimal
valuation and have a maximal valuation. Dropping this assumption does
not change the result in any essential way: some of the inequalities
are changed to strict inequalities.

} Then all the $\LEXNE$ of $G$ are of one of the following classes:
\begin{enumerate}
\item \textbf{No-transaction equilibria} (These equilibria exist for any
two sets $V_{s}$ and $V_{b}$) 

In these equilibria, neither the buyer nor the seller participates
(i.e., they play $\bot$, or bid a too extreme bid for all types of
the other player), regardless of their valuations.
\item \textbf{One-price equilibria} (These equilibria are defined only when
$\ls\leqslant\hb$, i.e., when an ex-post transaction is possible.) 

In a one-price equilibrium, both the seller and the buyer choose to
participate for some of their types. It is defined by a price $p\in\left[\ls,\hb\right]$
s.t. the equilibrium strategies are:
\begin{enumerate}
\item[] The seller bids $p$ for valuations $v_{s}\leqslant p$, and $\bot$
otherwise (the second clause might be vacuously true).
\item[] The buyer bids $p$ for valuations $v_{b}\geqslant p$, and $\bot$
otherwise (the second clause might be vacuously true).
\end{enumerate}
Hence, the outcome is\qquad{}%
\begin{tabular}{@{}r|c|c|}
\backslashbox{Buyer}{Seller}  &
Low:\quad{}$v_{s}\leqslant p$  &
High:\quad{}$v_{s}>p$\tabularnewline
\hline 
Low:\quad{}$v_{b}<p$  &
no transaction  &
no transaction\tabularnewline
\hline 
High:\quad{}$v_{b}\geqslant p$  &
$p$  &
no transaction\tabularnewline
\hline 
\end{tabular}.
\item \textbf{Two-price equilibria} (These equilibria are defined only when
 $\ls\leqslant\lb$ and $\hs\leqslant\hb$, i.e., when there is a
value for the seller s.t. an ex-post transaction is possible for any
value of the buyer, and vice versa.) 

In a two-price equilibrium, all types of both the seller and the buyer
choose to participate,  and their bids depend on their valuations.
It is defined by two prices $p_{L}<p_{H}$ s.t. $\left\{ \begin{array}{l}
\ls\leqslant p_{L}<\hs\leqslant p_{H}\\
p_{L}\leqslant\lb<p_{H}\leqslant\hb
\end{array}\right.$ and the equilibrium strategies are:
\begin{enumerate}
\item[] The seller bids $p_{L}$ for valuations $v_{s}\leqslant p_{L}$,
and $p_{H}$ otherwise.
\item[] The buyer bids $p_{H}$ for valuations $v_{b}\geqslant p_{H}$,
and $p_{L}$ otherwise.
\end{enumerate}
Hence, the outcome is\qquad{}%
\begin{tabular}{@{}r|c|c|}
\backslashbox{Buyer}{Seller}  &
Low:\quad{}$v_{s}\leqslant p_{L}$  &
High:\quad{}$v_{s}>p_{L}$\tabularnewline
\hline 
Low:\quad{}$v_{b}<p_{H}$  &
$p_{L}$  &
no transaction\tabularnewline
\hline 
High:\quad{}$v_{b}\geqslant p_{H}$  &
$x\left(p_{L},p_{H}\right)\in\left(p_{L},p_{H}\right)$  &
$p_{H}$\tabularnewline
\hline 
\end{tabular}. 
\end{enumerate}

\end{lem}
\begin{proof-sketch}
(The full details of the proof can be found in Appendix~\ref{Appendix:Proof of lem:Bilateral Trade - Main Result})\MyLyxThmNewline{} 

It is easy to verify that these profiles are indeed equilibria. We
will prove that they are the only equilibria. 

Assume for contradiction there is another $\LEXNE$ profile of actions
(i.e., bids or $\bot$), namely, $\left(a_{s},a_{b}\right)$. We define
$P_{s}$ to be the set of bids that are bid by the seller, i.e., $P_{s}=\left\{ a_{s}\left(v_{s}\right)\SetSt a_{s}\left(v_{s}\right)\neq\bot\right\} $,
and similarly we define $P_{b}=\left\{ a_{b}\left(v_{b}\right)\SetSt a_{s}\left(v_{b}\right)\neq\bot\right\} $.
Both these sets are not empty since $\left(a_{s},a_{b}\right)$ is
not a no-transaction equilibrium. First, we notice that if $P_{s}$
is of size one, i.e., whenever the seller participates he announces
$p$; then, if the buyer chooses to participate (based on his valuation),
he chooses to match $p$ in order to minimize the price (and vice
versa). Analyzing the valuations for which they choose to participate
proves that this profile is a one-price equilibrium. Now, assume that
both these sets are of size at least two. If both players participate
regardless of their value (never choose $\bot$), then the worst and
best cases for a player are those of facing the highest and lowest
bids of the other player. Hence his best response will be to match
one of the two, and we get a two-price equilibrium. If the seller
chooses whether to participate based on his value, i.e., there is
a value $v_{s}$ for which he chooses $\bot$, then the buyer cannot
guarantee himself more than zero (for instance, if he meets $v_{s}$).
Hence, he will choose one of the actions that guarantee him zero in
the worst case (e.g., $\bot$), and choose among these actions according
to their best case (meeting the lowest-bidding type of the seller).
Hence, given his value, the buyer either chooses $\bot$, or matches
the lowest bidding type of the seller. This proves this equilibrium
is either a no-transaction equilibrium or a one-price equilibrium.
The case in which the buyer chooses whether to participate based on
his value is symmetrical.
\end{proof-sketch}
We find it interesting that the set of equilibria depends on the possible
types of the players, and not on the price mechanism $x\left(a_{s},a_{b}\right)$.
In addition to the two classic equilibrium kinds, no-transaction equilibrium
and one-price equilibrium, we get a new kind of equilibrium. We see
that in this equilibrium each of the players announces one of two
bids, which is tantamount to  announcing whether his value is above
some threshold or not. This decision captures the (non-probabilistic)
trade-off a player is facing: whether to trade for sure, i.e., with
all types of the other player, or to get a better price. For example.,
the buyer decides whether to bid the high price and buy the item for
sure, taking the risk of paying more than his value; or whether to
bid the low price and buy at a lower price, taking the risk of not
buying at all. Since $\LEX$ is a function of the worst-case and best-case
outcomes only, it does not seem surprising that we get this dichotomous
trade-off and at most two bids (messages) for each player in equilibrium. 

This result might explain the emergence of market scenarios in which
a participant needs to choose which one of two markets to attend,
e.g., florists who choose whether to sell in a highly competitive
auction or in an outside market, and he needs to choose between the
two while not knowing the demand for that day. In a continuation work,
we follow this story, and analyze double auctions with several buyers
and sellers.

\section{Coordination games}

\label{sec:Coordination-Games}
In this section we study a second application of $\LEXNE$ to an economic
scenario: analyzing coordination games with type ambiguity. These
games model scenarios in which the participants prefer to coordinate
their actions with others, e.g., due to positive externalities.Some
examples are choosing a meeting place (a generalization the Battle
of the Sexes game~\cite[Ch.~5, Sec.~3]{Luce1957}), choosing a cellular
company, and placing a public good or bad when the cost is shared.
 We analyze coordination games in which all players prefer to maximize
the number of other players they coordinate with (while being indifferent
about their identity), but they might differ in their tie-breaking
rule between two maximizing actions. 
\begin{defn}[Coordination games with type ambiguity]
\MyLyxThmFN{Most of the current literature (e.g., \cite[pp.~54--74]{Schelling1980},
\cite[pp.~90--91]{Luce1957}, \cite[pp.~15--16]{Osborne1994}) deal
with two-player coordination games (games in which the best response
of a player is to copy the other player's action). We use the same
name here for multi-player (generalized) coordination games, which
capture the same kind of scenarios.}\MyLyxThmNewline{} A (finite) coordination game of $n$ players over
$m$ \emph{locations} is a game in which all players have the same
set of actions of size $m$ (and we refer to the actions as \emph{locations}),
and the preference of each player (his type) over the action profiles
is defined by a strict (ordinal) preference $\alpha$ over the locations
in the following way: Player~$i$, holding a preference $\alpha$
over the locations, strongly prefers an action profile $a=\left(a^{1},\ldots,a^{n}\right)\in\ActionSet$
to an action profile $b=\left(b^{1},\ldots,b^{n}\right)\in\ActionSet$
if either he meets more players under $a$ than under $b$ ($\sizeof{\left\{ j\neq i\SetSt a^{j}=a^{i}\right\} }>\sizeof{\left\{ j\neq i\SetSt b^{j}=b^{i}\right\} }$)
or if he meets the same (non-zero) number of players under both profiles
and he prefers the meeting location in $a$ to the one in $b$ ($a^{i}$
is preferred to $b^{i}$ according to $\alpha$). That is, the set
of types of player~$i$, $\TypeSet^{i}$, is a set of strict preferences
over the $m$ locations. In particular, a player is indifferent between
the outcomes in which he does not meet any of the other players. 
\end{defn}
First, we note that the concept $\MINNE$ is a too coarse for analyzing
coordination games; almost all pure action profiles of a (large enough)
coordination game are $\MINNE$.
\begin{lem}
\MyLyxThmNewline{}\label{lem:Almost all profiles are MINNE}
Let $G$ be a coordination game over $m$ locations with $n$ players,
each of them having at least $t$ types. Then more than $\left(1-\sfrac{1}{m^{t-1}}\right)^{n}$-fraction
of the pure action profiles are $\MINNE$ of the game.

Specifically, a profile is a $\MINNE$ of $G$ iff either there exists
a location $l$ s.t. all types of every player choose $l$ in the
profile, or there does not exist a location $l$ and a player $i$
s.t. all the types of player $i$ choose $l$ in the profile.

\end{lem}
We analyze the refinement of $\MINNE$, pure $\LEXNE$,\footnote{Since any $\LEXNE$ is also a $\MINNE$ this can be interpreted as
an equilibrium selection process.} which we show always exists. We show that every $\LEXNE$ profile
$a$ is uniquely defined by the set of locations chosen in $a$,\VspaceBeforeMath{-.7em}{}
\[
L\left(a\right)=\left\{ l\SetSt\exists i\in\Players,\,t^{i}\in\TypeSet^{i}\st\text{player\,}i\text{ of type }t^{i}\text{ plays }l\text{ in the profile }a\right\} .
\]

\begin{lem}
\MyLyxThmFN{This lemma is proved in Appendix~\ref{Appendix:Proof of lem:Coord General LEX characterization}}\label{lem:Coord General LEX characterization}\MyLyxThmNewline{}
Let $G$ be a coordination game and let $L$ be a non-empty set of
locations. There exists a $\LEXNE$ profile $a$ s.t. $L\left(a\right)=L$
if and only if for all $i$ the mapping $f^{i}:\TypeSet^{i}\rightarrow L$
that maps a type to his best location in $L$ is onto.

Moreover, the action profile $a$ in which every type of every player
chooses his best location in $L$ is the unique pure $\LEXNE$ that
satisfies $L\left(a\right)=L$.

\end{lem}

Abusing notation, we say that a location set $L$ is a $\LEXNE$ (and
shortly $L\in\LEXNE\left(G\right)$) if there exists a $\LEXNE$ profile
$a$ s.t. $L\left(a\right)=L$. Two immediate corollary from the lemma
are that there is always a pure $\LEX$ equilibrium and that any location
set $L\in\LEXNE$ satisfies $\sizeof L\leqslant\min_{i}\sizeof{\TypeSet^{i}}$,
and, in particular, if there is no ambiguity about the type of at
least one player ($\exists i\st\sizeof{\TypeSet^{i}}=1$), then the
only $\LEXNE$ are those in which all types of all players choose
the same location. Notice that for any vector of type sets $\left\{ \TypeSet^{i}\right\} $
an action profile in which all types of all players choose the same
location, which does not depend of the players' preferences, is always
a $\LEXNE$. In these profiles, any deviation of a (type of a) player
results in the deviator not meeting any of the other players  and
hence it is not beneficial for the deviator. 

From now on we assume that $\sizeof{\TypeSet^{i}}>1$ for all $i$
(there is a real ambiguity about the preference of each of the players),
and study the equilibria (location sets) according to $\LEX$. We
analyze the non-trivial equilibria that emerge due to type ambiguity.
We call a profile $p$ a \emph{non-trivial profile }if $\sizeof{L\left(p\right)}>1$,\footnote{I.e., there exist at least two different locations, each of which
is chosen by some type of some player.} and call a location set a \emph{non-trivial location set} if it includes
at least two locations.

From the characterization above, we prove several easy properties
of $\LEXNE$.

\begin{cor}
\MyLyxThmNewline{}Let $G$ be a coordination game over $m$ locations
and $n$ players with type sets $\TypeSet^{i}$.
\begin{itemize}
\item Let $a$ be an equilibrium profile; then the set\VspaceBeforeMath{-.7em}{}
\[
\left\{ l\SetSt\exists t^{i}\in\TypeSet^{i}\st\text{player }i\text{ of type }t^{i}\text{ plays }l\text{ in the profile }a\right\} 
\]
 is independent of $i$ (i.e., in equilibrium all players choose the
same set of actions).
\item (\emph{Increasing ambiguity}) Let $G'$ be a coordination game over
$m$ locations and $n$ players with type sets $\widehat{\TypeSet^{i}}$
satisfying $\TypeSet^{i}\subseteq\widehat{\TypeSet^{i}}$ for all
$i$. Then, $\LEXNE\left(G\right)\subseteq\LEXNE\left(G'\right)$.\footnote{In particular, the equilibria of the case with no ambiguity are also
equilibria of any coordination game.}
\item (\emph{$\LEXNE$ is downward closed}) If $L\in\LEXNE\left(G\right)$,
then any $L'\subseteq L$ is also an equilibrium of $G$. 
\item (\emph{Irrelevant information})
\begin{itemize}
\item Let $L$ be an equilibrium of $G$ and let $\widehat{\TypeSet^{i}}$
be the result of changing the preferences of the types of player~$i$,
while keeping the preferences over the locations in $L$. Then $L$
is also an equilibrium of the coordination game over $m$ locations
and $n$ players with type sets $\left\langle \TypeSet^{-i},\,\widehat{\TypeSet^{i}}\right\rangle $.\footnote{That is, the same type sets for all players except player~$i$, and
the perturbed type set for player~$i$.}
\item Let $L$ be an equilibrium of $G$ and let $G'$ be an extension of
$G$ to $m+1$ locations s.t. for any player and any of his types
the preference over the first $m$ locations is the same as in $G$.
Then $L$ is also an equilibrium of $G'$.
\end{itemize}
\end{itemize}
\end{cor}

\subsection{Coordination games with single-peaked consistent preferences}

Next, we study the equilibria of coordination games under type ambiguity
for several special cases in which the type sets satisfy some natural
constraints. In this work, we present cases in which for every player
the set of preferences $\TypeSet^{i}$ is single-peaked consistent
with regard to a line. 
\begin{defn}[single-peaked consistent preferences with regard to a line~\cite{Black1948}]
\MyLyxThmNewline{}

 A preference $\alpha$ over a set of locations $S\subseteq\Re$
is said to be \emph{single-peaked with regard to $\Re$}, if there
exists a utility function $f\colon\Re\rightarrow\Re$ s.t. for any
two locations $y$ and $z$, $y$ is preferred to $z$ if $f\left(y\right)>f\left(z\right)$
(i.e., $f$ represents $\alpha$); the top-ranked location  $x^{\star}$
in $\alpha$ is the unique maximizer of $f$; and for any two locations
$y$ and $z$, if $x^{\star}<y<z$ or $x^{\star}>y>z$, then $y$
is preferred to $z$.

A set of preferences over a set of locations $S$ is \emph{single-peaked
consistent w.r.t. a line, if} there exists an embedding function $e\colon S\rightarrow\Re$
(which we refer to as the \emph{order} of the locations) s.t. for
any preference $\alpha$ in the set, $e\left(\alpha\right)$ is single-peaked
with regard to $\Re$.
\end{defn}
For ease of presentation, we state the results for the case $S\subseteq\Re$
(and so the embedding is the identity function). For example,  consider
the scenario of deciding on locating a common good on a linear street.
It is known that each player holds an ideal location (his location),
and that his preference is monotone in the path between this ideal
location and the common good. Yet, there might be ambiguity about
the ideal location, or about the preference between locations that
are not on the same side of the ideal location (e.g, the preference
might be a function of properties of the path).

We interpret the single-peakedness assumption as constraining the
type ambiguity in a natural way: Player~$i$ knows the order over
the locations common to all types of player~$j$, and hence he has
some information he can use to anticipate the actions of other players.
 Note that a set of single-peaked consistent preferences might be
single-peaked with regard to more than one order of the locations
(i.e., an embedding).\footnote{\hspace*{\fill}\hspace*{\fill}\begin{minipage}[t]{.94\columnwidth}\vspace{-.5em}%
\begin{example}
\label{exa:several orders}\global\long\def\ORDER#1#2#3#4#5{#1\!\!-\!\!#2\!\!-\!\!#3\!\!-\!\!#4\!\!-\!\!#5}
\global\long\def\PREF#1#2#3#4#5{#1\SprefOver#2\SprefOver#3\SprefOver#4\SprefOver#5}
The preferences $\PREF 12345$, $\PREF 21345$, $\PREF 23145$, and
$\PREF 32145$ are single-peaked w.r.t. the following four orders
(and their inverse): $\ORDER 12345$, $\ORDER 51234$, $\ORDER 51234$,
and $\ORDER 54123$. 
\end{example}
\end{minipage}

Escoffier et al.~\cite{Escoffier2008} proved that for any number
of locations $n$ and $r\leqslant2^{n-1}$, there exist  $\frac{1}{r}2^{n-1}$
different preferences that are single-peaked consistent w.r.t. $r$
different orders (for the tight bounds see \cite{Escoffier2008}).

} In cases where $\TypeSet^{i}$ is single-peaked consistent w.r.t.
several orders of the locations (which can be thought of as a stronger
 constraint), our results hold w.r.t. any of the orders, giving rise
to stricter characterizations. In the two examples  to follow, we
limit the ambiguity further (i.e., adding more information on players)
in two ways.

The first case we analyze is when assuming no ambiguity about the
players' ideal locations. That is, for every player there exists a
location $x^{\star}$ (his ideal location) that is shared by all his
types in $\TypeSet^{i}$. We show that in this case there can be at
most two locations in the equilibrium set, one situated to the right
of all ideal locations and one situated to the left of all of them
(\enquote{right} and \enquote{left} w.r.t. the order of each player).

\begin{lem}
\label{lem:Coordination Line SP  peak known} \MyLyxThmNewline{} 

Let $G$ be a coordination game with single-peaked consistent type
sets $\TypeSet^{i}$ s.t. there is no ambiguity about the players'
ideal locations, and let $x^{i}$ be player~$i$'s ideal location.
Then any (non-trivial) $\LEXNE$ $L$ satisfies that there exist two
different locations $\alpha$ and $\beta$ s.t. $L=\left\{ \alpha,\beta\right\} $
and for every player~$i$ and order $<^{i}$ s.t. $\TypeSet^{i}$
is single-peaked consistent w.r.t $<^{i}$ it holds that $\alpha<^{i}x^{i}<^{i}\beta$.

This condition is tight. Any set $L$ satisfying the above condition
is a $\LEXNE$ of the game in case the type sets $\TypeSet^{i}$ are
rich enough.\footnote{In particular, if for all players the type set $\TypeSet^{i}$ contains
all single-peaked preferences with ideal location $x^{i}$ w.r.t.
to some order (and only them), then the above characterizes all (non-trivial)
$\LEXNE$.}

In particular, in case \textup{$\cup_{i}\TypeSet^{i}$ is single-peaked
consistent w.r.t. an order $<$, }any (non-trivial) $\LEXNE$ $L$
satisfies that there exist two different locations $\alpha$ and $\beta$
s.t. $L=\left\{ \alpha,\beta\right\} $ and $\alpha<\min_{i}x^{i}\leqslant\max_{i}x^{i}<\beta$.
\end{lem}
\begin{proof}
\MyLyxThmNewline{}

Notice that the result is equivalent to requiring that for every player~$i$
there be at most one location $\alpha$ in $L$ s.t. $\alpha<^{i}x^{i}$
and at most one location $\beta$ in $L$ s.t. $x^{i}<^{i}\beta$.
Assume for contradiction that either there are two locations in $L$
to the left of $x^{i}$ or two locations to the right of $x^{i}$.
Then one of them must be on the path between $x^{i}$ and the other
location, and hence preferred to the other location regardless of
the type of player~$i$, in contradiction to $L$ being an equilibrium.

In order to prove tightness, given a set $L=\left\{ \alpha,\beta\right\} $
as above, for any player~$i$ there are single-peaked preferences
with $x^{i}$ at the top in which $\alpha$ is preferred to $\beta$,
and there are single-peaked preferences with $x^{i}$ at the top in
which $\beta$ is preferred to $\alpha$. If the set $\TypeSet^{i}$
includes both a type that prefers $\alpha$ to $\beta$ and a type
that prefers $\beta$ to $\alpha$, then the mapping $f^{i}:\TypeSet^{i}\rightarrow L$
of a type to his preferred location in $L$ is onto. Hence, if this
condition is satisfied by all type sets, then $L$ is a $\LEXNE$.
\end{proof}
We see that under this homogeneity constraint, the (non-trivial) equilibria
are constrained to be two extreme locations. This result can also
be interpreted as the power of extreme players in scenarios where
there is a common agreement on the order. Note that when the types
are single-peaked w.r.t. several orders, the equilibria should satisfy
the above characterization w.r.t. all of them. For instance, consider
the scenario of two players and type sets $\TypeSet^{1}=\TypeSet^{2}=\left\{ \PREF 21345\,,\,\PREF 23145\right\} $
(a subset of the set of preferences in Example~\ref{exa:several orders}).
These preferences are single-peaked w.r.t. the order $\ORDER 51234$,
but $L=\left\{ 4,5\right\} $ is not a $\LEXNE$, albeit satisfying
the property of the lemma because both types prefer $4$ to $5$ and
will not choose $5$, and indeed $\left\{ 4,5\right\} $ does not
satisfy the property w.r.t. the order $\ORDER 12345$. 

The second case we analyze is having ambiguity about the ideal location
only, while assuming a structure on the preferences (in a sense, this
case is the complement of the first case). A simple example of such
a structure is that of Euclidean preferences. A Euclidean preference
is uniquely derived from its ideal location by ordering the locations
according to their distance from the ideal location. Notice that this
restriction too can be stated as a homogeneity constraint on the preferences
given a common embedding: it can be stated as a common metric on the
location set that is shared by all preferences.

\begin{lem}
\label{lem:Coordination Line SP  Euclidean} \MyLyxThmNewline{} 

Let $G$ be a coordination game on the real line s.t. all players
have Euclidean preferences. Then a (non-trivial) location set $L=\left\{ l_{1}<l_{2}<\cdots<l_{k}\right\} $
is a $\LEXNE$ if and only if   for every player $i$ there are
types derived by possible ideal locations for him $x_{1}<x_{2}<\cdots<x_{k}$
s.t. for $t=1,2,\ldots,k-1$\VspaceBeforeMath{-.7em}{}
\[
x_{t}<\dfrac{l_{t}+l_{t+1}}{2}<x_{t+1},
\]
i.e., the median between $l_{t}$ and $l_{t+1}$ is between $x_{t}$
and $x_{t+1}$.\footnote{In the proof we do not use the fact that all type sets $\TypeSet^{i}$
are Euclidean w.r.t. the same embedding, and actually the condition
in the lemma can be weakened to hold only with respect to each of
the embeddings of $\TypeSet^{i}$. } 
\end{lem}

A special case of the lemma is where there are exactly two types of
each player that are derived from two possible ideal locations $x^{i}<y^{i}$.
In this scenario we get that a (non-trivial) location set $L$ is
a $\LEXNE$ if and only if there exist two locations $\alpha$ and
$\beta$ s.t. $L=\left\{ \alpha,\beta\right\} $ and $m=\half\left(\alpha+\beta\right)$
is (strictly) between $x^{i}$ and $y^{i}$ for all players. Hence,
a (non-trivial) $\LEXNE$ exists if and only if the intersection of
the segments $\left(x^{i},y^{i}\right)$ is non-empty.
\begin{proof}
\MyLyxThmNewline{}

$\underline{\Rightarrow}$: Let $L=\left\{ l_{1}<l_{2}<\cdots<l_{k}\right\} $
be an equilibrium and let $i$ be a player. Since $L$ is an equilibrium,
there are types of player~$i$,\footnote{We identify the types with their ideal locations.}
$x_{1}<x_{2}<\cdots<x_{k}$, s.t. player~$i$ chooses $l_{t}$ when
his type is $x_{t}$. In particular, for $t<k$ his preference between
$l_{t}$ and $l_{t+1}$ when his type is $x_{t}$ is different from
his preference when his type is $x_{t+1}$. Hence, $x_{t}$ and $x_{t+1}$
lie on different sides of the median between $l_{t}$ and $l_{t+1}$.

$\underline{\Leftarrow}$: Following the same reasoning, it is easy
to see that if $L$ satisfies the property of the lemma, then player~$i$
of type $x_{t}$ prefers $l_{t}$ to any other location in $L$, and
hence $L$ is an equilibrium.
\end{proof}

\section{Summary and future directions}

 How people choose an action to take when possessing only partial
information, and how they should choose their action, are basic questions
in economics, and  on them the definition of equilibrium is built
(both as a prediction tool and as a self-enforcing contract).  The
main stream of game-theoretic literature assumes that economic agents
are expectation maximizers (according to some objective or subjective
prior) and, moreover, that there is some consistency between the players'
priors (commonly, the common prior assumption). 

In this work, we chose the  inverse scenario and studied cases in
which the players have no information on the state of the world. 
We defined a general framework of games with ambiguity following Harsanyi's
model of games with incomplete information~\cite{Harsanyi1967},
and  in particular, games with type ambiguity. We axiomatized a family
of decision models under ambiguity that we claim a rational agent
is expected to follow, and characterized the finest refinement of
this family, $\LEX$.  This family can be interpreted as all \emph{rational}
models of decision in cases of (extreme) ambiguity that follow Wald's
MiniMax principle, defined by the different  ways to decide in cases
in which the MiniMax principle is mute. In many scenarios with type
ambiguity Wald's MiniMax principle is too coarse, and so the corresponding
equilibrium ($\MINNE$) has almost no predictive power. The way to
justify a selection process from these equilibria is to refine the
players' preference, that is, assume they act according to a decision
model in this family. We showed $\LEX$ is the unique model that follows
Wald's MiniMax principle and breaks all possible instances of indifference
(without violating the rationality axioms we assumed).  Finally,
we studied the respective equilibrium notion, $\LEXNE$, and applied
it to two families of games: coordination games and bilateral trade
games.

 One might ask himself why to choose $\LEX$ as the analysis tool,
and not a different decision model in the family. First, we note that,
just like the $\MIN$ model, the $\LEX$ model has a simple and intelligible
cognitive interpretation, and so it does not require a complex epistemological
assumption on the players, which other models might impose. In addition,
  we note that $\LEXNE\left(G\right)$, i.e., the equilibria of
$G$ under the $\LEX$, are in fact also equilibria of $G$ under
any profile of rational refinements of Wald's MiniMax principle. Moreover,
$\LEXNE\left(G\right)$ can be equivalently defined as the set of
all such robust equilibria  of $G$. For instance, these are the
profiles an outside actor can suggest as a self-enforcing contract,
even if he does not know the exact preferences of the players.

This scenario of extreme ambiguity  might seem unrealistic. Yet,
we claim this model approximates many partial-information real-life
scenarios better than the subjective expectation maximization model.
Clearly, if players have information which they can use to construct
a belief about the other players, and we expect they will use it,
then the expectation maximization decision model is a better analysis
tool. In intermediate scenarios, when players have some information
but it is unreasonable to expect them to form a distribution over
the world, it is reasonable to model the players as following one
of the intermediate models for decision making under ambiguity,\footnote{We refer to them as intermediate since they do not satisfy the principle
of indifference, and hence they differentiate between the states of
the world.} e.g., the multi-prior model (for an overview of such models, the
interested reader is referred to \cite{Gilboa2011}). It remains an
open question  what are the common features of equilibria under
$\LEX$ and equilibria under other decision models under ambiguity.
By characterizing these common features, we would like to analyze
the sensitivity of the equilibria we've found in this work to the
specific analysis tool.  One could also justify the model we presented,
via justifying the axioms, for other scenarios. E.g., for scenarios
in which the players might have some information on the preferences
of others, but due to extreme risk aversion or bounded rationality
constraints, they follow Wald's MiniMax model. Moreover, in cases
in which one justified the axioms we presented (and mostly the invariance
to irrelevant information axiom), e.g., on cognitive grounds, we get
that the $\LEXNE$ is the right analysis tool for the same reasons
we presented above.

In order to study further  the notion of equilibrium under $\LEXNE$,
we hope to analyze other families of games, e.g., finer cases of coordination
games when adding homogeneity constraints that are either knowledge
on other players, or intra-player agreement. In addition, we see
few further directions of research.

\uline{Variance of information between types}: In the examples
we analyzed the information of a player did not depend on his type
(in other words, a player cannot deduce from his type any information
on the feasible types of others). The model we presented in Sec.~\ref{sec:Model}
includes  more general scenarios. We have preliminary results, which
we omitted here, for Schelling's Homeowner-Burglar game~\cite[p.~207]{Schelling1980}.
 For this game, we got similar predictions to the prediction of Schelling,
replicating the power of partial-knowledge of high degrees.

\uline{Mechanism design}: In a continuation work, we extend the
result regarding bilateral trade, to characterization of incentive
compatible, individually rational, deterministic mechanisms for bilateral
trade. We show that (essentially) the implementable allocation rules,
are those that are implemented using the price announcement mechanisms
we analyzed. Also here, implementability under $\LEXNE$ can be interpreted
as robust implementability under ambiguity, i.e., implementability
without assuming a specific decision model, but analyzing the profiles
that are equilibria for any profile of decision models of the players
from the family we characterized. Another direction of research is
analyzing what a designer can gain by adding ambiguity to a mechanism,
and specifically whether it is possible to increase the participation
of the players, similarly to the full-participation result we proved.

\uline{Information update}: The main drawback of modeling the knowledge
of a player using a set of types for other players (or, in the general
case, of states of the world) instead of a richer structure, is that
there is no reasonable way to define information update in this model.
This prevents us to extend this work to two natural directions: analysis
of extensive form games (e.g., when the players play in turns, learning
the type of each other as the evolve), and value of information (the
question of how much a player should invest in order to decrease his
ambiguity). In the decision theory literature, there are several non-Bayesian
information update rules, e.g., Dempster-Shafer~\cite{Dempster1967,Shafer1976}
 and Jeffrey update rule~\cite{Jeffrey1965}. These rules usually
assume a finer representation of knowledge than the representation
we had in this work, but we think that after basing a rational decision
model in our simplified knowledge representation, it should not be
hard to extend the decision model to these finer knowledge models.

\phantomsection


\newpage{}

\appendix

\section{Proof of Lemma~\ref{lem:BRexists} (Best response is well defined)}

\label{Appendix:Proof of lem:BRexists}This lemma can be easily extended
to infinite number of states case by assuming some structure on the
actions set and the utility function.
\begin{lem*}
\MyLyxThmNewline{}
The following best response correspondences are non-empty: pure $\MINBR$,
mixed $\MINBR$, pure $\LEXBR$, and mixed $\LEXBR$.

\end{lem*}
\begin{proof}
\MyLyxThmNewline{}
We prove the lemma in the more general decision theory framework.
We show that if there are finitely many states of the world and finitely
many pure actions, there is always at least one optimal mixed action.
Since the number of players and the number of types of each player
are finite, a player faces one of a finite number of profiles (states
of the world) and we get the desired result.

The existence of an optimal pure action is trivial since there are
finitely many pure actions. Next, we prove the existence of optimal
mixed actions. \global\long\def\OptAct{\operatorname{OptAct}}
 The set of optimal actions according to $\MIN$ $\OptAct_{\MIN}$
is\VspaceBeforeMath{-.7em}{}
\[
\OptAct_{\MIN}=\argmax_{\sigma\in\Delta\left(\ActionSet\right)}\min\nolimits _{\omega\in\Omega}u\left(\sigma,\omega\right).
\]

For every state of the world $\omega\in\Omega$, $u\left(\sigma,\omega\right)$
is a continuous function in $\sigma$ (it is a linear transformation).
$\min\nolimits _{\omega\in\Omega}u\left(\sigma,\omega\right)$ is
also continuous in $\sigma$ as the minimum for finitely many continuous
functions. Hence, $\argmax_{a\in\sigma\in\Delta\left(\ActionSet\right)}\min\nolimits _{\omega\in\Omega}u\left(\sigma,\omega\right)$
is non-empty as the maxima of a continuous function over a simplex.
Moreover, $\OptAct_{\MIN}$ is a compact set. 

Similarly, the set of optimal actions according to $\LEX$\VspaceBeforeMath{-.7em}{}
\[
\OptAct_{\LEX}=\argmax_{a\in\OptAct_{\MIN}}\max\nolimits _{\omega\in\Omega}u\left(\sigma,\omega\right)
\]
is non-empty since $\OptAct_{\MIN}$ is compact and $\max\nolimits _{\omega\in\Omega}u\left(\sigma,\omega\right)$
is continuous in $\sigma$. 

\end{proof}

\section{Proof of Theorem~\ref{thm:NEexists} ($\protect\MINNE$ existence)}

\label{Appendix:Proof of thm:NEexists}
\begin{thm*}
\MyLyxThmNewline{}
Every game with ambiguity has a $\MINNE$ in mixed actions.

\end{thm*}
\begin{proof}
\MyLyxThmNewline{}
\global\long\def\SET{\mathbb{S}}
We prove the existence of a mixed $\MINNE$ by applying Kakutani's
fixed point theorem~\cite{Kakutani1941}\footnote{Let $\SET\subseteq\Re^{n}$ be a non-empty compact convex set. Let
$F:\SET\rightarrow\SET$ be a set-valued upper semi-continuous function
on $\SET$ such that $F\left(s\right)$ is non-empty and convex for
all $s\in\SET$. Then $F$ has a fixed point, i.e., a point $s\in\SET$
such that $s\in F\left(s\right)$.} to the following set $\SET$ and set-valued function $F$. We set
$\SET\subseteq\Re^{K}$ to be the set of all profiles of mixed strategies
of the types,\footnote{for $K$ being the sum of the number of types over the players.}
i.e., a cross product of the corresponding simplexes. Given a strategy
profile $s$, we define $F\left(s\right)$ to be the product of the
best response correspondences of the different types to the profile
$s$.

Clearly, $\SET$ is a convex compact non-empty set. We proved (Lemma~\ref{lem:BRexists})
that the best response always exists and hence $F\left(s\right)$
is non-empty for all $s\in\SET$. We show that the best response correspondence
of a type is an upper semi-continuous function (and hence also $F$)
by applying Berge's Maximum theorem~\cite[Thm.~2, p.~116]{Berge1963}\footnote{If $f\colon X\times Y\rightarrow\Re$ is a continuous function, then
the mapping $\mu\colon X\rightarrow Y$ defined by $\mu\left(x\right)=\argmax_{y\in Y}f\left(y\right)$
is a upper semi-continuous mapping. \\
(The statement of the theorem is taken from \cite[Lemma~17.31,~p.~570]{Aliprantis2006})}  to (where $t^{-i}\left(\omega\right)$ are the types player~$i$
is facing according to $\omega$, and $s^{-i}\left(t^{-i}\left(\omega\right)\right)$
are their actions according to $s$)\VspaceBeforeMath{-0.5em}{} 
\[
\mbox{f\ensuremath{\left(\sigma^{i},s\right)}=}\min_{\omega\in\Omega}u\left(\omega,\sigma^{i},s^{-i}\left(t^{-i}\left(\omega\right)\right)\right).
\]
We notice that the best response of a type of player~$i$ to a strategy
profile $s\in\SET$ is the maxima over his mixed strategies $\sigma^{i}$
of $\mbox{f\ensuremath{\left(\sigma^{i},s\right)}}$;  that $u$
is a linear function in $\sigma^{i}$ and in $s$; and that $f\left(\sigma^{i},s\right)$
is a continuous function in $s$ and $\sigma^{i}$. Hence, we get
that the best response is an upper semi-continuous correspondence. 

Next, we show that for any profile $s$ and a type of player~$i$,
the best response set is convex, and hence also $F\left(s\right)$
is a convex set. In order to prove the convexity of the best response
set, let $\sigma$ and $\tau$ be two best responses and let\VspaceBeforeMath{-0.5em}{}
\[
m=\min_{\omega\in\Omega}u\left(\omega,\sigma^{i},s^{-i}\left(t^{-i}\left(\omega\right)\right)\right)
\]
be their worst-case value. For any convex combination $\zeta$ of
$\sigma$ and $\tau$, on one hand we get that \VspaceBeforeMath{-0.5em}{}
\[
\min_{\omega\in\Omega}u\left(\omega,\zeta^{i},s^{-i}\left(t^{-i}\left(\omega\right)\right)\right)\leqslant m
\]
from the optimality of $\sigma$, and on the other hand \VspaceBeforeMath{-0.5em}{}
\[
u\left(\omega,\zeta^{i},s^{-i}\left(t^{-i}\left(\omega\right)\right)\right)\geqslant m\text{ for all }\omega\in\Omega.
\]
since we assumed this lower-bound both on $\sigma$ and on $\tau$.
Hence, \VspaceBeforeMath{-0.5em}{}
\[
\min_{\omega\in\Omega}u\left(\omega,\zeta^{i},s^{-i}\left(t^{-i}\left(\omega\right)\right)\right)=m
\]
and $\zeta$ is a best response.

By applying Kakutani's fixed point theorem we get that there is a
profile $p$ s.t. $p\in F\left(p\right)$. That is, in the profile
$p$ every type best-responds to the others, so $p$ is a $\MINNE$. 

\end{proof}

\section{Proof of Lemma~\ref{lem:MIN and MAX only} (Axiomatization)}

\label{Appendix:Proof of lem:MIN and MAX only}
\begin{lem*}
\MyLyxThmNewline{}
Any preference that satisfies monotonicity, state symmetry, and independence
of irrelevant information can be defined as a function of the worst
and the best outcomes of the actions.

\end{lem*}
\begin{proof}
\MyLyxThmNewline{}
For this proof it is easier to use the following property that is
equivalent (when $\Omega$ is finite) to state symmetry.
\begin{ax*}
\MyLyxThmNewline{}For any action $a$ and bijection $\psi\colon\Omega\rightarrow\Omega$,
the DM is indifferent between $a$ and $a\circ\psi$.
\end{ax*}
Let $a$ and $a'$ be two actions over the states set $\Omega$ s.t.
$\min_{\omega\in\Omega}a\left(\omega\right)=\min_{\omega\in\Omega}b\left(\omega\right)$
and $\max_{\omega\in\Omega}a\left(\omega\right)=\max_{\omega\in\Omega}b\left(\omega\right)$.
We show that the DM is indifferent between the two, and notice that
this will prove the lemma. Assuming this claim, given any two pairs
of actions: $a$ and $a'$ that have the same minimal and maximal
outcome, and $b$ and $b'$ that have the same minimal and maximal
outcome, $a$ is preferred to $b$ if and only if $a'$ is preferred
to $b'$, because the DM is indifferent between $a$ and $a'$, and
between $b$ and $b'$.

Assume for contradiction that (w.l.o.g) the DM strictly prefers $a'$
to $a$. Since the preference satisfies state symmetry, we can assume
that $a$ and $a'$ are co-monotone, that is, for any two states $\omega$
and $\omega'$, the actions satisfy that $\left(a\left(\omega\right)-a\left(\omega'\right)\right)\left(a'\left(\omega\right)-a'\left(\omega'\right)\right)\geqslant0$.
We denote by $m$ and $M$ the minimal and maximal outcomes of $a$
and by $\omega_{m}$ and $\omega_{M}$ the two respective states of
world.\VspaceBeforeMath{-.7em}{}
\[
\begin{array}[b]{c|ccc}
 & \omega_{m} & \text{other states} & \omega_{M}\\
\hline a & m & \in\left[m,M\right] & M\\
a' & m & \in\left[m,M\right] & M
\end{array}
\]

We define two new actions over $\Omega$. An action $b$ that results
in $m$ in all states besides $\omega_{M}$ and $M$ otherwise, and
an action $b'$ that results in $M$ in all states besides $\omega_{m}$
and $m$ otherwise.\VspaceBeforeMath{-.7em}{}
\[
\begin{array}[b]{c|ccc}
 & \omega_{m} & \text{other states} & \omega_{M}\\
\hline a & m & \in\left[m,M\right] & M\\
b & m & m & M\\
a' & m & \in\left[m,M\right] & M\\
b' & m & M & M
\end{array}.
\]

Due to monotonicity the DM (weakly) prefers $a$ to $b$, and $b'$
to $a'$, and hence the DM strictly prefers $b'$ to $b$. Next, we
define an auxiliary space $\widehat{\Omega}=\left\{ \omega_{m},\omega_{o},\omega_{M}\right\} $
and two actions on it $c$ and $c'$ by unifying the middle states
into one as follows:\VspaceBeforeMath{-.7em}{}
\[
\begin{array}[b]{c|ccc}
 & \omega_{m} & \omega_{o} & \omega_{M}\\
\hline c & m & m & M\\
c' & m & M & M
\end{array}.
\]
Since the preference satisfies Independence of Irrelevant Information
we get that the DM strictly prefers $c'$ to $c$. Now we define a
new action $c''$ (see below), and by the state symmetry property
we get that also $c''$ is strictly preferred to $c$. \VspaceBeforeMath{-.7em}{}
\[
\begin{array}[b]{c|ccc}
 & \omega_{m} & \omega_{o} & \omega_{M}\\
\hline c & m & m & M\\
c' & m & M & M\\
c'' & M & M & m
\end{array}.
\]

Using a collapsing argument similar to the one above, we get that
for the following collapsed space  and two actions\VspaceBeforeMath{-.7em}{}
\[
\begin{array}[b]{c|cc}
 & \omega_{m} & \omega_{M}\\
\hline d & m & M\\
d'' & M & m
\end{array},
\]
 $d''$ is strictly preferred to $d$, but this contradicts the state
symmetry property. 

\end{proof}

\section{Proof of Proposition~\ref{prop:GS+mine=00003DMINexp} (Axiomatization
of $\protect\MIN$)}

\label{Appendix:Proof of prop:GS+mine=00003DMINexp}
\begin{prop*}
\MyLyxThmNewline{}

Let $X$ be a set of outcomes and let $Y$ be the set of distributions
over $X$ with finite supports, and we identify $X$ with the Dirac
distributions $\left\{ y\in Y\SetSt y\left(x\right)=1\text{ for some }x\in X\right\} $.

Let $\Omega$ be a finite set of states of the world s.t. $\sizeof{\Omega}\geqslant3$
and let $L=Y^{\Omega}$ (i.e., actions which in each state of the
world return a lottery over $X$). For $f,g\in L$ and $\alpha\in\left[0,1\right]$,
we define $\alpha f+\left(1-\alpha\right)g$ as the state-wise compound
lottery, that is, the action that returns for each state $\omega\in\Omega$
the lottery $\alpha f\left(\omega\right)+\left(1-\alpha\right)g\left(\omega\right)\in Y$. 

We will denote by $L_{c}\subset L$ the constant functions in $L$,
and by $L_{o}=X^{\Omega}\subset L$ the functions that return pure
outcomes.

Let $\prefOver$ be a preference order over $L$ that satisfies Gilboa-Schmeidler
Axioms~$A.2$-$A.5$~\cite{Gilboa1989} 
\begin{itemize}
\item \emph{$A.2$: Certainty independence:} For all $f,g$ in $L$ and
$h$ in $L_{c}$ and for all $\alpha\in\left(0,1\right)$:\VspaceBeforeMath{-.7em}{}
\[
f\SprefOver g\iff\alpha g+\left(1-\alpha\right)h\SprefOver\alpha f+\left(1-\alpha\right)h.
\]
\item \emph{$A.3$: Continuity:} For all $f$, $g$ and $h$ in $L$: if
$f\SprefOver g$ and $g\SprefOver h$ then there are $\alpha$ and
$\beta$ in $\left(0,1\right)$ such that $\alpha f+\left(1-\alpha\right)h\SprefOver g$
and $g\SprefOver\beta f+\left(1-\beta\right)h$.
\item \emph{$A.4$: Monotonicity:} For all $f$ and $g$ in $L$: if $f\left(\omega\right)\prefOver g\left(\omega\right)$
for all $\omega\in\Omega$ then $f\prefOver g$.
\item \emph{$A.5$: Uncertainty aversion:} For all $f,g\in L$ and $\alpha\in\left(0,1\right)$:
$f\prefOver g$ and $g\prefOver f$ implies $\alpha f+\left(1-\alpha\right)g\prefOver f$.
\end{itemize}
and satisfies state symmetry and independence of irrelevant information
on $L_{o}$.

Then $\prefOver$ is represented by a function $u\colon X\rightarrow\Re$
such that for any $f,g\in Y$\VspaceBeforeMath{-.7em}{} 
\[
f\prefOver g\,\iff\,\min_{\omega\in\Omega}\Exp_{x\sim f\left(\omega\right)}\left[u\left(x\right)\right]\geqslant\min_{\omega\in\Omega}\Exp_{x\sim g\left(\omega\right)}\left[u\left(x\right)\right].
\]

\end{prop*}
\begin{proof}
\MyLyxThmNewline{}

Let $\prefOver$ be a preference order which satisfies the above.
Gilboa and Schmeidler~\cite{Gilboa1989} proved that $\prefOver$
is represented by a function $u\colon X\rightarrow\Re$ and a non-empty,
closed and convex set $C$ of probability measures on $\Omega$ such
that for any $f,g\in L$\VspaceBeforeMath{-.7em}{} 
\[
f\prefOver g\,\iff\,\min_{P\in C}\Exp_{\omega\sim P}\left[\Exp_{x\sim f\left(\omega\right)}\left[u\left(x\right)\right]\right]\geqslant\min_{P\in C}\Exp_{\omega\sim P}\left[\Exp_{x\sim g\left(\omega\right)}\left[u\left(x\right)\right]\right]
\]
and in particular for any $a,b\in L_{o}$\VspaceBeforeMath{-.7em}{}
\[
a\prefOver b\,\iff\,\min_{P\in C}\Exp_{\omega\sim P}\left[u\left(a\left(\omega\right)\right)\right]\geqslant\min_{P\in C}\Exp_{\omega\sim P}\left[u\left(b\left(\omega\right)\right)\right].
\]

Let $l$ and $h$ be two outcomes s.t. the DM strictly prefers the
action that gives him $h$ in all states of the world over the action
that gives him $l$ in all states of the world, and so $u\left(h\right)>u\left(l\right)$.
We denote $u\left(h\right)$ and $u\left(l\right)$ by $u_{h}$ and
$u_{l}$, respectively.

For every $\omega\in\Omega$, we define $\epsilon_{\omega}=\min_{P\in C}P\left(\omega\right)$,
and we define the actions $l_{\omega}$ and $h_{\omega}$ by $l_{\omega}\left(\omega'\right)=\begin{cases}
l & \omega=\omega'\\
h & \text{otherwise}
\end{cases}$ and $h_{\omega}\left(\omega'\right)=\begin{cases}
h & \omega=\omega'\\
l & \text{otherwise}
\end{cases}$. From Lemma~\ref{lem:MIN and MAX only} we get that the preference
on $L_{0}$ can be defined as a function of the worst and the best
outcomes of the actions, and so the DM is indifferent between $l_{\omega^{\star}}$
and $h_{\omega^{\dagger}}$ for every $\omega^{\star},\omega^{\dagger}\in\Omega$.
Hence,\VspaceBeforeMath{-.7em}{} 
\[
u_{l}+\left(h_{h}-u_{l}\right)\sum_{\omega\neq\omega^{\star}}\epsilon_{\omega}\leqslant\,\min_{P\in C}\Exp_{\omega\sim P}\left[u\left(l_{\omega^{\star}}\left(\omega\right)\right)\right]=\min_{P\in C}\Exp_{\omega\sim P}\left[u\left(h_{\omega^{\dagger}}\left(\omega\right)\right)\right]\,=u_{l}+\left(u_{h}-u_{l}\right)\epsilon_{\omega^{\dagger}},
\]
and $\sum_{\omega\neq\omega^{\star}}\epsilon_{\omega}\leqslant\epsilon_{\omega^{\dagger}}$.
Hence, since $\sizeof{\Omega}>2$ we get that $\epsilon_{\omega}=0$
for all $\omega\in\Omega$.

Next we prove that for any action $a\in L_{o}$ it holds that\VspaceBeforeMath{-.7em}{}
\[
\min_{P\in C}\Exp_{\omega\sim P}\left[u\left(a\left(\omega\right)\right)\right]=\min_{\omega\in\Omega}u\left(a\left(\omega\right)\right).
\]
 The DM is indifferent between $a$ and the action $a'$ that gives
him $\argmax_{\omega\in\Omega}u\left(a\left(\omega\right)\right)$
on one state $\omega^{\star}$ and $\argmin_{\omega\in\Omega}u\left(a\left(\omega\right)\right)$
on all other states. Hence,\VspaceBeforeMath{-.7em}{}
\[
\begin{array}{rl}
\min_{P\in C}\Exp_{\omega\sim P}\left[u\left(a\left(\omega\right)\right)\right] & =\min_{P\in C}\Exp_{\omega\sim P}\left[u\left(a'\left(\omega\right)\right)\right]\\
 & =\min_{\omega\in\Omega}u\left(a\left(\omega\right)\right)+\left(\max_{\omega\in\Omega}u\left(a\left(\omega\right)\right)-\min_{\omega\in\Omega}u\left(a\left(\omega\right)\right)\right)\cdot\min_{P\in C}P\left(\omega^{\star}\right)\\
 & =\min_{\omega\in\Omega}u\left(a\left(\omega\right)\right).
\end{array}
\]

We get that for any action $a\in L_{0}$ it holds that $\min_{P\in C}\Exp_{\omega\sim P}\left[u\left(a\left(\omega\right)\right)\right]=\min_{\omega\in\Omega}\left[u\left(a\left(\omega\right)\right)\right]$,
and so the Dirac distribution concentrated on $\omega$ belongs to
$C$ for every state of the world $\omega\in\Omega$ . Hence for any
$f,g\in Y$:\VspaceBeforeMath{-.7em}{}
\[
\begin{array}[b]{rl}
f\prefOver g & \iff\min_{P\in C}\Exp_{\omega\sim P}\left[\Exp_{x\sim f\left(\omega\right)}\left[u\left(x\right)\right]\right]\geqslant\min_{P\in C}\Exp_{\omega\sim P}\left[\Exp_{x\sim g\left(\omega\right)}\left[u\left(x\right)\right]\right]\\
 & \iff\min_{\omega\in\Omega}\Exp_{x\sim f\left(\omega\right)}\left[u\left(x\right)\right]\geqslant\min_{\omega\in\Omega}\Exp_{x\sim g\left(\omega\right)}\left[u\left(x\right)\right].
\end{array}\qedhere
\]

\end{proof}

\section{Proof of Lemma~\ref{lem:Bilateral Trade - Main Result} (Bilateral
trade games)}

\label{Appendix:Proof of lem:Bilateral Trade - Main Result}
\begin{lem*}
\MyLyxThmNewline{}
\end{lem*}
\begin{proof}
\MyLyxThmNewline{}We will prove that these profiles are equilibria
($\LEXNE$) and that they are the only equilibria.
\begin{itemize}
\item \uline{No-transaction profiles}

In these profiles for at least one of the players does not participate
($\bot$, or bids an extreme bid) regardless of his type.

A player, facing a profile in which all the types of the other player
do not participate, is indifferent between all his actions (all of
them result in no-transaction for sure), and hence he best-responds.
Facing any other profile, at least for some of his types (e.g., the
maximal value for the buyer, and minimal value of the seller), choose
to participate. Hence, the only no-transaction profile that is an
equilibrium, is when both players do not participate.
\item \uline{One-price profiles}

In these profiles, both the seller and the buyer choose to participate
for some of their types, and when they choose to participate, both
bid the same price $p$.

\medskip{}

First, we prove that given a one-price profile as described in the
lemma, it is indeed an equilibrium.

Consider a seller of type $v_{s}$ facing a profile in which there
are types of the buyer that choose $p$ and (maybe) others that choose
not to participate.
\begin{itemize}
\item If $v_{s}\leqslant p$: He prefers bidding $p$ to any bid $a_{s}<p$,
because both bids will result in a transaction with the same types
of the buyer, but $a_{s}$ will result in a lower price (Monotonicity
of $\LEX$).

He prefers bidding $p$ to any bid $a_{s}>p$, because $a_{s}$ will
result in no-transaction with all types of the buyer, while $p$ results
in a profitable transaction with some of them (Monotonicity of $\LEX$).

We notice that since $p\geqslant\ls$ there are types of the seller
that bid $p$.
\item If $v_{s}>p$: Any bid that results in a transaction with some types
of the buyer, will be at a price lower than $v_{s}$, and hence the
transaction will be a losing one. Hence, he prefers not to participate
at all (Monotonicity of $\LEX$).
\end{itemize}
Similarly, consider a buyer of type $v_{b}$ facing a profile in which
there are types of the seller that choose $p$ and (maybe) others
that choose not to participate.
\begin{itemize}
\item If $v_{b}<p$: Any bid that results in a transaction with some types
of the seller, will be at a price higher than $v_{b}$, and hence
the transaction will be a losing one. Hence, he prefers not to participate
at all (Monotonicity of $\LEX$).
\item If $v_{b}\geqslant p$: He prefers bidding $p$ to any bid $a_{b}>p$,
because both bids will result in a transaction with the same types
of the seller, but $a_{s}$ will result in a higher price (Monotonicity
of $\LEX$).

He prefers bidding $p$ to any bid $a_{s}<p$, because $a_{s}$ will
result in no-transaction with all types of the seller, while $p$
results in a profitable transaction with some of them (Monotonicity
of $\LEX$).

We notice that since $p\leqslant\hb$ there are types of the buyer
that bid $p$.
\end{itemize}
\medskip{}

Next, we prove that any single-price profile other than those described
in the lemma, cannot be an equilibrium.

First we claim, that in equilibrium it cannot be that the prices of
participating types of one of the sides vary, while all types of the
other player bid the same. If there is a unique price $p$ bid by
participating types of the seller, based on the analysis above, all
participating types of the buyer bid $p$. Similarly, if there is
a unique price $p$ bid by participating types of the buyer, all participating
types of the seller bid $p$.

Hence, in a single-price profile, there exists a price $p$ s.t. both
some types of the seller and some types of the buyer choose $p$,
and all other types (might be an empty set) choose not to participate.

If $p\in\left[\ls,\hb\right]$, by the analysis above, we see there
is only one single-price equilibrium supporting $p$.

If $p\notin\left[\ls,\hb\right]$, by the analysis above, either all
types of the buyer or all types of the seller choose not to participate.
In either case this is a no-transaction profile.
\item \uline{All-participating profiles}

In these profiles, all types of both the seller and the buyer choose
to participate (did not choose $\bot$). 

\medskip{}

First we prove that given an all-participating profile as described
in the lemma as two-price equilibria, it is indeed an equilibrium.

Consider a seller of type $v_{s}$ facing a profile in which all the
types of the buyer participate, and let $p_{L}$ be the infimum of
the bids chosen by types of the buyer and $p_{H}$ the supremum of
the bids.
\begin{itemize}
\item If $v_{s}\leqslant p_{L}$: He prefers bidding $p_{L}$ to any bid
$a_{s}<p_{L}$, because both bids will result in a transaction with
all types of the buyer, but $a_{s}$ will result in a lower price
with each of them (Monotonicity of $\LEX$).

He prefers bidding $p_{L}$ to any bid $a_{s}>p_{L}$, because $a_{s}$
will result in no-transaction with some types of the buyer, while
$p_{L}$ results in a non-losing transaction with all of them ($\LEX$
is a refinement of $\MIN$).
\item If $v_{s}>p_{L}$ and $v_{s}\leqslant p_{H}$: He prefers bidding
$p_{H}$ to any bid $a_{s}<p_{H}$, because when facing the types
of the buyer that bid $a_{b}\in\left[a_{s},p_{H}\right]$, both bids
will result in a transaction, but $a_{s}$ will result in a lower
price, and when facing other types of the buyer (that bid less than
$a_{s})$, both result in no-transaction (Monotonicity of $\LEX$).

He prefers bidding $p_{H}$ to any bid $a_{s}>p_{H}$, because $a_{s}$
results in no-transaction with all types of the buyer, while $p_{H}$
results in a non-losing transaction with some types of the buyer (Monotonicity
of $\LEX$).
\end{itemize}
We notice that since in the profiles of the lemma satisfy $\ls\leqslant p_{L}<\hs\leqslant p_{H}$,
all the types of the seller participate, there are types of the seller
that bid $p_{L}$, and there are types that bid $p_{H}$.

Similarly, consider a buyer of type $v_{b}$ facing a profile in which
all the types of the seller participate, and let $p_{L}$ be the infimum
of the bids chosen by types of the seller and $p_{H}$ the supremum
of the bids.
\begin{itemize}
\item If $v_{b}<p_{H}$ and $v_{b}\geqslant p_{L}$: He prefers bidding
$p_{L}$ to any bid $a_{b}<p_{L}$, because $a_{b}$ results in no-transaction
with all types of the seller, while $p_{L}$ results in a non-losing
transaction with some types of the seller (Monotonicity of $\LEX$).

He prefers bidding $p_{L}$ to any bid $a_{b}>p_{L}$, because when
facing the types of the seller that bid $a_{s}\in\left[p_{L},a_{b}\right]$,
both bids will result in a transaction, but $a_{b}$ will result in
a higher price, and when facing other types of the seller (that bid
more than $a_{b})$, both result in no-transaction (Monotonicity of
$\LEX$).
\item If $v_{b}\geqslant p_{H}$: He prefers bidding $p_{H}$ to any bid
$a_{b}<p_{H}$, because $a_{b}$ will result in no-transaction with
some types of the seller, while $p_{H}$ results in a non-losing transaction
with all of them ($\LEX$ is a refinement of $\MIN$).

He prefers bidding $p_{H}$ to any bid $a_{b}>p_{H}$, because both
bids will result in a transaction with all types of the seller, but
$a_{b}$ will result in a higher price with each of them (Monotonicity
of $\LEX$). 
\end{itemize}
We notice that since in the profiles of the lemma satisfy $p_{L}\leqslant\lb<p_{H}\leqslant\hb$,
all the types of the buyer participate, there are types of the buyer
that bid $p_{L}$, and there are types that bid $p_{H}$.

\medskip{}

Next, we prove that any all-participating non single-price  profile,
other than the above, cannot be an equilibrium.

In an all-participating non single-price  profile, all the types of
both the buyer and seller participate, and for both players, not all
their types choose the same price.

If there is a type of the buyer whose value is lower than any of the
prices bid by types of the seller, then this type prefers not to participate,
so it cannot be an equilibrium. Any other type of the buyer bids either
the supremum of the infimum of the prices bid by types of the seller.
Given that types of the buyer bid one of two prices, if there is a
type of the seller whose value is higher than both prices, then this
type prefers not to participate, and the profile is not an equilibrium.
Otherwise, the types of the seller also bid these two prices and we
get the profile is of the type described in the lemma.
\item Last, we show that all other profiles are not equilibria of the game.

Assume towards a contradiction that $\overline{a}$ is an equilibrium
profile of the game that is not one of the above profile types. 

Clearly, for a match of any pair of types of the seller and the buyer,
the utility of neither of them cannot be negative, because he can
guarantee himself to get at least zero against all types of the other
player by choosing $\bot$.

Assume there is a type $v_{s}$ of the seller that chose $\bot$.
Then for any participating type $v_{b}$ of the buyer, $v_{b}$ cannot
guarantee himself to get a utility greater than zero (his utility
from being matched with $v_{s}$) in the worst case. Hence, he breaks
the tie between all the bids $a_{b}\leqslant v_{b}$, all giving him
zero in the worst case, according to the best case which is being
matched with the lowest bid of a participating type of the seller.
Hence, we get a contradiction by showing the profile is a one-price
profile.

Similarly, assume there is a type $v_{b}$ of the buyer that chose
$\bot$. Then for any participating type $v_{s}$ of the seller, $v_{s}$
cannot guarantee himself to get a utility greater than zero (his utility
from being matched with $v_{b}$) in the worst case. Hence, he breaks
the tie between all the bids $a_{b}\geqslant v_{s}$, all giving him
zero in the worst case, according to the best case which is being
matched with the highest bid of a participating type of the buyer.
Hence, we get a contradiction by showing the profile is a one-price
profile.\qedhere
\end{itemize}
\end{proof}

\section{Proof of Lemma~\ref{lem:Coord General LEX characterization} (Coordination
games)}

\label{Appendix:Proof of lem:Coord General LEX characterization}
\begin{lem*}
\MyLyxThmNewline{}
Let $G$ be a coordination game and let $L$ be a non-empty set of
locations. There exists a $\LEXNE$ profile $a$ s.t. $L\left(a\right)=L$
if and only if for all $i$ the mapping $f^{i}:\TypeSet^{i}\rightarrow L$
that maps a type to his best location in $L$ is onto.

Moreover, the action profile $a$ in which every type of every player
chooses his best location in $L$ is the unique pure $\LEXNE$ that
satisfies $L\left(a\right)=L$.

\end{lem*}
\begin{proof}
\MyLyxThmNewline{}
The case $\sizeof L=1$ is trivial. Let $l$ be the location in $L$.
The only profile satisfying $L\left(a\right)=L$ is when all players
choose $l$ regardless of their type, and this is an equilibrium.
The mappings $f^{i}$ are onto, and hence the lemma is proved for
this case. 

From now on, we assume that $\sizeof L\geqslant2$. Given a profile
$a$ we define $L^{i}\left(a\right)$ to be the set of locations chosen
in $a$ by player~$i$, 
\[
L^{i}\left(a\right)=\left\{ l\SetSt\exists t^{i}\in\TypeSet^{i}\st\text{player\,}i\text{ of type }t^{i}\text{ plays }l\text{ in the profile }a\right\} ,
\]

so $L\left(a\right)=\cup_{i\in\Players}L^{i}\left(a\right)$. 

$\underline{\Rightarrow}$: 

Assume there exists an equilibrium $a$ s.t. $L\left(a\right)=L$. 

First we claim that for all players, $\sizeof{L^{i}\left(a\right)}>1$.
Assume towards a contradiction, there is a non-empty set of players
$S$ s.t. $i\in S\iff\sizeof{L^{i}\left(a\right)}=1$, and let $p$
be a player in $S$ and $l$ the location chosen by $p$. 

If $\sizeof S=1$, then for any of the other players the only action
guaranteeing them to meet another player is choosing $l$, and hence
all players choose $l$ regardless of their type. I.e., $L\left(a\right)=\left\{ l\right\} $
and we get a contradiction.

If $\sizeof S>1$, then any of the players can guarantee himself to
meet another player on the worst case, and hence will choose a location
in $\cup_{i\in S}L^{i}\left(a\right)$. Since all players aim to maximize
the number of other players they meet on the worst case, it cannot
be that two players in $S$ choose differently (both maximize the
number of other players from $S$ that choose like them). Hence, $\sizeof{\cup_{i\in S}L^{i}\left(a\right)}=1$
and we get that $\sizeof{L\left(a\right)}=1$. Contradiction.

Hence, for all players, $\sizeof{L^{i}\left(a\right)}>1$. Player~$i$
of type $t^{i}\in\TypeSet^{i}$ cannot guarantee himself to meet other
players, and hence he will choose the best location according to his
type among the locations in $\cup_{j\neq i}L^{j}\left(a\right)$ that
maximize the number of other players he might meet - $m\left(l\right)=\#\left\{ j\neq i\SetSt l\in L^{j}\left(a\right)\right\} $. 

Next, we prove that for all players $L^{i}\left(a\right)=L\left(a\right)$.
Assume for contradiction there exists  Player~$p$ s.t. $L^{p}\left(a\right)\subsetneq L\left(a\right)$.
Let $l\in L^{p}\left(a\right)$ be a location chosen by $p$, and
let $l'$ be a location and $p'$ a player s.t. $l'\in L^{p'}\left(a\right)\setminus L^{p}\left(a\right)$.
Since Player~$p$ (of some type) chose $l$, we get that the number
of players $j\neq p$ that chose $l$ is at least as large as the
number of players $j\neq p$ that chose the location $l'$. Taking
the viewpoint of Player~$p'$ we get that the number of players $j\neq p'$
that chose $l$ is strictly larger than the number of players $j\neq p'$
that chose the location $l'$. But this is a contradiction to $l'\in L^{p'}\left(a\right).$

Hence, we get that in $a$, player~$i$ of type $t^{i}\in\TypeSet^{i}$
choose the best location according to his type among the locations
in $L$ (in all of them he meets no-one in the worst case and everyone
else in the best case). That is, player~$i$ plays according to the
function $f^{i}$, and $L=L^{i}=Im\left(f\right)$, so $f$ is onto.
Notice that we also get this profile is the unique equilibrium s.t.
$L\left(a\right)=L$.

$\underline{\Leftarrow}$: 

Assume the functions $f^{i}$ are onto. It is easy to verify that
in the profile $a$ in which player~$i$ plays according to $f^{i}$,
player~$i$ best-responds (for all $i$). None of his actions guarantees
him to meet some other player in the worst case; if he plays an action
$a\notin L$, he does not meet another player even in the best case;
if he plays an action $a\in L$, he meets all other players in the
best case; hence, he best-responds playing according to $f^{i}$. 

\end{proof}

\end{document}